\documentclass[12pt]{article}
\usepackage{setspace}
\usepackage{amsmath,amsthm,amsfonts,amssymb}
\usepackage[utf8]{inputenc}
\usepackage[english]{babel}
\usepackage{geometry}
\usepackage[colorinlistoftodos]{todonotes}
\usepackage{lmodern}
\usepackage{booktabs}
\usepackage{makecell}
\usepackage{tabularx}
\usepackage{graphicx}
\graphicspath{{figures/}{./}}
\usepackage[hypertexnames=false]{hyperref}
\usepackage{float}
\usepackage[nottoc,notlot,notlof]{tocbibind}
\usepackage{natbib}
\usepackage{authblk}
\usepackage{enumitem}
\usepackage{bm}
\usepackage{csquotes}
\usepackage{rotating}
\usepackage{multirow}

\newtheorem{assumption}{Assumption}
\newtheorem{lemma}{Lemma}
\newtheorem{proposition}{Proposition}

\title{
Ranked-choice conjoint experiments\thanks{
We would like to thank Don Green, Kosuke Imai, Ben Lauderdale, Sara Hobolt, Michael Laver, Sarah Brierley, and Andy Eggers, as well as participants at the 2024 European Political Science Association annual conference and LSE Political behavior seminar, for their helpful comments on this project. Both empirical studies described in this paper underwent ethics review at the London School of Economics. 
The pre-analysis plans for this study are registered at \url{https://osf.io/fawuy}, and the \texttt{cjrank} R package is available at \url{https://github.com/tsrobinson/cjrank}.}}%

\author{Thomas S. Robinson\thanks{Department of Methodology, London School of Economics.\\ (email: \texttt{t.robinson7@lse.ac.uk})} \hspace{10mm}
Mats Ahrenshop\thanks{Department of Politics and International Relations, University of Oxford.\\(email: \texttt{mats.ahrenshop@politics.ox.ac.uk})} \hspace{10mm}
  Spyros Kosmidis\thanks{Department of Politics and International Relations, University of Oxford.\\(email: \texttt{spyros.kosmidis@politics.ox.ac.uk})} 
}
\date{April 2026}

\begin{document}
\onehalfspacing
\hypersetup{pageanchor=false}

\maketitle

\begin{abstract}
Forced-choice conjoint designs have become a staple method in the experimentalist's toolkit. However, the forced-choice outcome is neither always consistent with the types of choices individuals make in real political contexts, nor is it statistically efficient. In this paper, we formalize how \textit{ranked} outcomes can be integrated into the conjoint framework. We provide a proof that rank-expanded estimators are equivalent to conventional AMCE, a theoretical account of how additional profiles increase the efficiency of conjoint designs, and design-based tests for the transitivity and independence of irrelevant alternatives assumptions that underpin the expansion. Across two pre-registered survey experiments--the first comparing forced-choice and ranked-choice designs across candidate and policy domains, and the second varying the number of ranked profiles--we find that ranked-choice conjoints yield substantively similar but more precise AMCE estimates, shrinking standard errors by 12-13\% with one additional profile and up to 55\% with six profiles per vignette. Based on efficiency--validity trade-offs, we recommend $K = 4$ profiles for most applications. We provide an accompanying open-source \texttt{R} package, \texttt{cjrank}, that implements rank expansion, AMCE estimation, efficiency diagnostics, and the assumption tests described in this paper. 

\end{abstract}

\thispagestyle{empty}

\newpage
\thispagestyle{empty}

\newpage

\doublespacing

\pagenumbering{arabic}
\hypersetup{pageanchor=true}

\section{Introduction}

Since their introduction to political science by \citet{Hainmueller2014}, conjoint experiments have become a staple of political scientists' research toolkit. In these designs, subjects are faced with a choice between (typically two) profiles. The ``conjoint" aspect of the design stems from the fact that each profile is described by multiple attributes, the values of which are randomized. Within political science, conjoint experiments have helped understand preferences over tax policy \citep{ballard2017structure}, voters' expectations over corruption \citep{spencer2020appearance}, and emigration decisions \citep{duch2022nativist}, to name a few.

The core benefits of the conjoint design stem from two primary features. First, conjoint experiments allow researchers to estimate treatment effects less intrusively than direct preference measures. By averaging across choices, profiles, and the values of other attributes, researchers can estimate the average marginal component effect (AMCE) of attribute-levels (relative to some set of baselines). Second, and unlike simpler survey experiments, the conjoint design allows multiple dimensions to be randomized at once--allowing researchers both to test multiple hypotheses more efficiently, and making the treatment better resemble the information contexts individuals likely face in real political situations \citep{Hainmueller2014}.

Notwithstanding these advantages, conjoint experiments are far from a silver bullet. Beyond criticisms over the choice of \textit{estimand}  \citep{abramson2022we}, coefficient interpretation \textit{across} models \citep{leeper2020measuring}, and measurement error \citep{clayton2023correcting}, conjoint experiments typically require subjects to complete many, repetitive rounds in order to obtain precise estimates of the marginal effects of interest \citep{Bansak2018}. In doing so, and even if survey satisficing remains low, conjoint experiments can be costly to run due to the number of rounds subjects have to consider.

In this paper, we formalize the integration of ranked outcomes into the conjoint framework for political science. Rankings have a longer history in marketing conjoint analysis \citep{vermeulen2011rank}, but their properties have not been established for the AMCE estimand central to political science applications, nor have formal diagnostics been developed. Our theoretical contributions are threefold. First, we prove the connection between the ranking and forced-choice estimands. Second, we show that efficiency gains scale with the number of profiles. Third, we introduce design-based empirical tests for transitivity and independence of irrelevant alternatives (IIA), the assumptions that underpin rank expansion. To the best of our knowledge, the ranked conjoint design has received next to no attention by political scientists, despite its appealing properties and natural suitability to many types of political decisions individuals make.

The ranked-choice design offers two complementary advantages. First, from a statistical perspective, ranking profiles allows researchers to recover more information more quickly. For a fixed number of randomized profiles, the ranked-choice design requires fewer rounds to be shown to subjects and, for each of those rounds, we recover more comparisons. We demonstrate that, for a fixed number of profiles, the resulting estimates will be more precise under a transitivity assumption. Second, from a measurement perspective, rankings capture more information from each vignette. The forced-choice measure discards all preference signals beyond the top choice, whereas rankings reveal subjects' relative preferences over all options. We show that these additional comparisons carry genuine preference content, rather than noise, against independent intensity ratings. This amounts to a purely within-respondent efficiency gain--more information per round from the same respondent--rather than the elicitation of a qualitatively different kind of preference data, and it is especially attractive as the field increasingly relies on probability samples where survey time is expensive \citep{Bansak2018}.

To assess the performance of ranked-choice conjoints, we conduct two pre-registered survey experiments.\footnote{Pre-registration documents can be accessed at 
\url{https://osf.io/fawuy}. The original plan submitted 06 June 2024 covers Study 1. A supplementary plan was submitted on 02 April 2025 to pre-register further hypotheses on a new sample. Both plans are available from the same link.} Study 1 compares forced-choice and ranked-choice designs across candidate election \citep{Kirkland18} and budget spending \citep[cf.][]{bansak2021} domains. Study 2 varies the number of ranked profiles ($K = 2, 4, 6$) and introduces bespoke tests of the transitivity and IIA assumptions. Across both studies, ranked-choice conjoints yield substantively similar but more precise AMCE estimates, with clustered standard errors 12-55\% smaller depending on $K$. These efficiency gains hold even when comparing designs that present the same total number of profiles, confirming that the advantage reflects how many pairwise comparisons each vignette yields rather than the total number of profiles evaluated. Ranked-choice data also enables better out-of-sample prediction of held-out profile choices \citep[building on][]{bansak2021}.

We accompany this paper with an open-source \texttt{R} package, \textbf{cjrank}, available at \url{https://github.com/tsrobinson/cjrank}. The package provides functions for expanding rank orderings into pairwise comparisons, estimating AMCEs with clustered standard errors, assessing efficiency across rounds, and testing for transitivity and IIA violations.

\section{Rankings as an alternative to forced choices}

While conjoint experiments have been used to study preferences since the 1970s, the method was popularized in political science by \citet{Hainmueller2014}. Subjects typically face a choice between two ``profiles", which are bundles of randomized treatments along a number of dimensions referred to as ``attributes". Each subject completes multiple rounds of these tasks, where each vignette re-randomizes the descriptions of each profile. For the sake of consistency, let a typical conjoint design be posed to $N$ individuals, who complete $J$ rounds, each of which has $K$ profiles, and each profile has $L$ attributes. The number of alternative attribute-levels within each attribute is denoted as $v_L$.

The outcome in conjoint experiments is measured at the subject-round-profile level. In the vignette, subjects are typically prompted with a question of the form ``If you had to choose between these two profiles,..." which is recorded using a binary indicator  (i.e., $y_{ijk} \in \{0,1\}$). Beyond the ability to test multiple preference dimensions simultaneously, therefore, one advantage of conjoint experiments is that for a single vignette the researcher recovers $K$ observations, in comparison to simpler randomized designs where one typically recovers a single observation per vignette.

Nothing prevents $K$ being equal to 1 or larger than 2. For example, \citet{duch2025governing} show a \textit{single} randomized profile per round. In single-profile designs, subjects choose whether or \textit{not} to support that specific profile. When $K>2$, the number of observations scales linearly with the number of profiles but power will not scale as quickly--for any two profiles not chosen by the respondent, the shared outcome is 0 and thus indiscriminate.

\subsection{Ranked outcomes}
\label{sec:rank}

To garner more information from more profiles, the experimenter needs to measure subjects' relative preferences over all options in the vignette. An intuitive way to do so is to ask subjects to rank profiles, rather than choose their most preferred. Rankings identify not only the top choice (as in the forced-choice measure) but also indicate subjects' relative preferences over all other profiles too. Ranking measurements are therefore denser in terms of the information they can capture compared to forced-choice outcomes. 

The lack of attention paid to implementing rankings within conjoint experiments is surprising, given how ubiquitous the act of ranking is within many of the political contexts we study \citep{atsusaka2022causal}. In studies of candidate elections, for example, and outside the minority of states that use majoritarian electoral systems, ranking candidates is a more natural analogue of the task that voters actually perform at the ballot box \citep{golder2005democratic,atsusaka_2023}. In policymaking contexts, too, lawmakers and citizens may be presented with numerous options, and it may be more natural to rank these outcomes than simply choose one over all others.

More broadly, rankings can improve intra-coder reliability relative to repeated pairwise comparisons \citep{kaufman2021measure} and may help address some of the measurement error inherent to forced-choice conjoint tasks \citep{clayton2023correcting}. Posing more than two alternatives can depress serial non-participation of respondents leading to more robust models \citep{rolfe2009impact}.

A further advantage of ranking is structural: because ranking extracts more information per round, the same number of total profiles can be evaluated in fewer rounds. This reduces exposure to potential carryover effects, whereby attribute values shown in earlier rounds influence how subjects evaluate profiles in later rounds \citep{Ham_Imai_Janson_2024}.

\subsection{Ranking profiles in a conjoint}

Let the $k$th profile shown to subject $i$ in round $j$ of a conjoint experiment be described as $p_{ijk}$. Each potential profile $p$ has a corresponding attribute vector $X_p$ describing the $L$ attribute values. We will assume that profiles are realized by independently and uniformly assigning attribute levels at random. The full set of profiles shown in that round is denoted $\mathbf{P}_{ij}$. For simplicity, we will often drop $j$ and focus on a single round of a conjoint experiment. 

Let $R_{i}(p_{ik}, \mathbf{P}_{i[-k]}) \in \{1,\ldots,K\}$ denote subject $i$'s potential ranking of the $k$th profile, given $K-1$ other profiles in that vignette, and let $R_{ik}$ be the realized rank for ${p}_{ik}$.  Finally, for comparison, let $Y_{ik}(p_{ik},p_{ik'}) \in \{0,1\}$ indicate whether subject $i$ would select the $k$th profile in a forced choice against $p_{ik'}$.

Taking measures $\mathbf{R}_{i}$, the ranked-choice outcome yields $K$ observations per vignette (the same as in a forced-choice vignette). We then ``expand" the measured rankings to recover pairwise comparisons between every combination of two profiles in that vignette. That is, for every possible pairwise comparison of profiles in a vignette, we assign the outcome $1$ to the profile with the higher rank, and $0$ to the other (Figure \ref{fig:rcc_schematic}).

\begin{figure}
    \centering
    \includegraphics[width=0.9\textwidth]{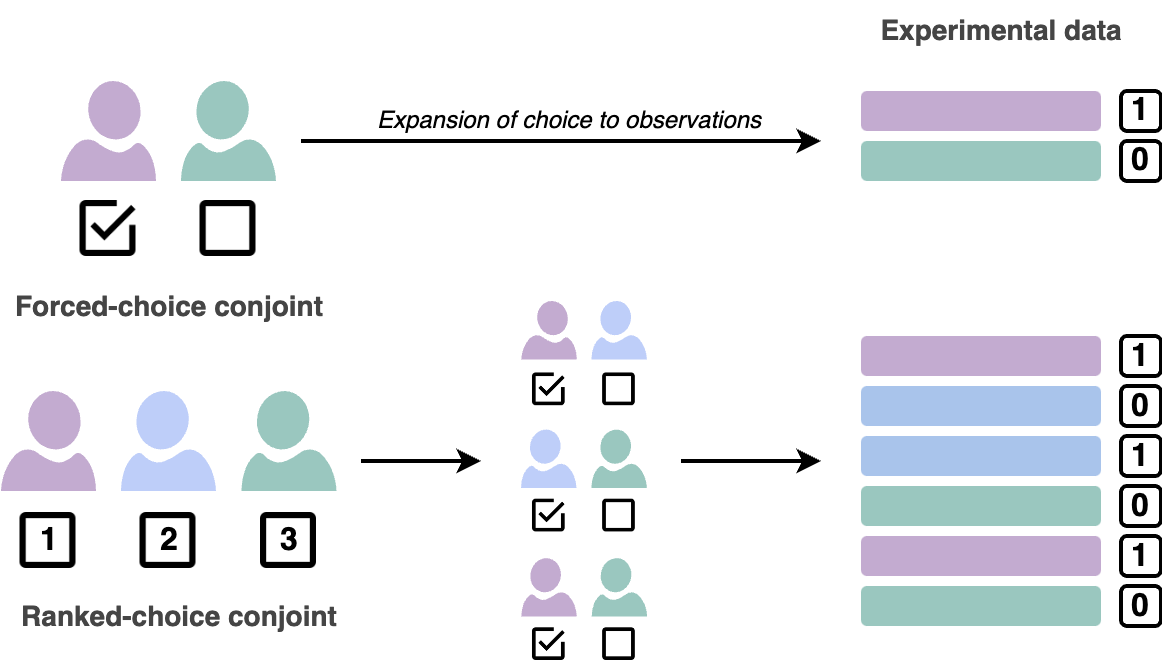}
    \caption{Difference in how forced-choice versus ranked-choice decisions expand to experimental data.}
    \label{fig:rcc_schematic}
\end{figure}

Since the expanded pairwise comparisons maintain the same structure as the forced-choice design, we employ exactly the same estimating strategy to recover an estimate of the AMCE. Thus, aside from adapting the vignettes, the researcher has to make few adaptations to their analysis framework.\footnote{Modeling the ranked-choice outcomes directly is possible, building on recent advances in incorporating ranking outcomes into the potential outcomes framework \citep{atsusaka2022causal}. We do not focus on this approach because, in doing so, we would necessarily shift the estimand away from probabilistic statements about ``choosing" one profile over another.}

That is, we recover estimates of the AMCEs by estimating a linear probability model of the form:
\begin{equation}
\label{eq:lpm}
y_{ijk} = \sum_{l}\bm{X_l}\bm{\beta_{l}} + e_{ijk},
\end{equation}
where $\bm{X_l}$ is the $N \times (v_L-1)$ indicator matrix of randomized treatments for the $l$th attribute, and $\bm{\beta_l}$ is the corresponding vector of AMCEs.

The statistical advantage of this expansion is the increase in power entailed by the larger number of observations. Whereas the forced-choice design scales at a constant rate--adding an extra profile in a vignette increases the number of observations by one--the ranked-choice design increases the number of comparisons non-linearly. For three profiles per vignette, we recover 3 comparisons and thus 6 observations; for four profiles we recover 6 comparisons and thus 12 observations; and for five profiles, 20 observations and so forth.\footnote{The number of observations is $2 \times \binom{K}{2}$.} In Section \ref{sec:sim}, we demonstrate by simulation the extent to which this increase in observations improves the power of the conjoint model.

\subsection{AMCE validity under rank expansion}

We must make two further assumptions alongside the standard assumptions laid out in \citet{Hainmueller2014} to guarantee the unbiasedness of the AMCE estimates derived from Equation \ref{eq:lpm}. Here, we briefly outline the intuition behind each. We prove this result formally in Appendix Section \ref{sec:proof}.

\paragraph{Assumption (A): Transitivity} \textit{
For any set of profiles $\{p, p', p''\}$ and latent utility function $U_i(\cdot)$, if $U_i(p) > U_i(p') > U_i(p'')$, then $R_{ip} < R_{ip'} < R_{ip''}$.
}

The first, and intuitive, assumption is that a subject's latent utilities over profiles are transitive. Under the random utility view that motivates conjoint designs, each profile’s attributes map onto a latent utility, and respondents rank profiles by comparing these utilities. Transitivity simply requires that this latent preference order is internally consistent: if a respondent prefers $p_x$ to  $p_y$ and  $p_y$ to  $p_z$, then they must prefer  $p_x$ to  $p_z$. Violations could arise from stochastic choice noise, inattentive responding, or measurement error when respondents convert utilities into an explicit ranking.

The assumption is most defensible when the task is cognitively light--for example, ranking three or four profiles. With few profiles it is unlikely that respondents will cycle or accidentally invert ranks. As $K$ grows, the cognitive challenge increases and maintaining a perfectly consistent global order becomes harder, so transitivity should be treated as an empirical question rather than a given. In the following section, we outline a test experimenters can include that estimates whether subjects' rankings are transitive.

\paragraph{Assumption (B): Pairwise independence (IIA)} \textit{
$\Pr(R_{ip} < R_{ip'} | X_p, X_p', X_r) = \Pr(R_{ip} < R_{ip'} | X_p, X_p')$.
}

The second, and perhaps less intuitive, assumption stems from the fact that one's relative ranking between any two profiles is made in the presence of other profiles. Therefore, in order to expand out this ranking to the hypothetical forced choice between \textit{only} those two profiles, we must assume that the \textit{relative} ranking of two profiles is not conditional on the descriptions of any other profile shown in that vignette ($X_r$). This pairwise independence assumption stems from IIA because, put differently, the presence of any other profiles $p''$ in the vignette must be irrelevant to the relative ranking of $p$ and $p'$. For example, a violation would occur if the partisan composition of the other $K-2$ profiles in a candidate vignette changed the weight a subject places on career experience enough to reverse the relative ranking of two focal profiles $p$ and $p'$.

In practice, two features of the conjoint design work in favor of this assumption. First, because attribute levels are independently and uniformly randomized, the ``remaining'' profiles $X_{-(pp')}$ are statistically independent of any given pair $(X_p, X_{p'})$. This means that even if IIA is violated--i.e., the remaining profiles influence pairwise comparisons--randomization prevents this influence from systematically correlating with specific attribute-levels, bounding the scope of any resulting bias. Second, we show formally (Appendix~\ref{sec:iia_violation}) that bias arises only through non-linear interactions between focal and non-focal profile attributes: if the influence of remaining profiles is approximately additive, OLS remains unbiased. In practice, this means that mild violations of IIA--where the presence of other profiles slightly shifts the weight a respondent places on a given attribute, but does not reverse pairwise rankings--will produce negligible bias. Empirically, we confirm this in Study 2: IIA violation rates remain below 25\% even at $K = 6$, and re-estimating AMCEs after excluding violators yields nearly identical results (Section~\ref{sec:assumptions}).

\subsection{Uncertainty estimation}

There are two inferential concerns related to the rank expansion. First, observations derived from the same vignette are not independent: under rank expansion, all $2\binom{K}{2}$ pairwise outcomes are deterministic functions of the same $K$ ranks. This within-vignette dependence does not affect the unbiasedness of the AMCE estimator, but means that standard errors computed under an i.i.d.\ assumption will be too small. Second, the effective sample size is $NJK$ profile-level observations rather than $2\binom{K}{2}NJ$ expanded rows, because each profile's attributes are repeated across $K-1$ pairwise comparisons (Proposition~\ref{prop:rankequiv}). Both concerns are addressed by clustering standard errors at the subject level, which accounts for all sources of within-subject dependence--including the mechanical correlation among rows derived from the same vignette and the accumulation of such dependence as $K$ grows.\footnote{Empirically, the ratio of subject-clustered to unclustered standard errors is approximately 1.05 for forced-choice, 1.22 for $K=3$, 1.43 for $K=4$, and 1.78 for $K=6$--confirming that clustering is increasingly important as $K$ grows. All standard errors in this paper use CR2 variance estimation.}

\section{Simulation evidence of ranking efficiency}
\label{sec:sim}

We run Monte Carlo simulations to assess the power advantage of ranking. We simulate a conjoint with six binary attributes, 500 subjects, and three rounds, drawing latent utilities from a linear model $U_{ijk} = \bm{\gamma}\bm{X}_{ijk} + \epsilon_{ijk}$ with $\epsilon \sim \mathcal{N}(0,1)$ and coefficients $\bm{\gamma}$ spanning a wide range.\footnote{$\bm{\gamma} = [0.01,0.02,0.04,0.06,0.08,0.1,0.15,0.2,0.25,0.5,1,2]$} The RCC condition uses $K = 3$ profiles per round; the FCC condition uses $K = 2$. Over 1,000 replications, Figure \ref{fig:power_difference} confirms that RCC yields higher power, with the largest gains at moderate effect sizes (0.06-0.2 SD of residual variance), corresponding to conventional AMCE magnitudes of about 0.02.\footnote{In Appendix Figure \ref{fig:amce_comparison}, we show that the densities of AMCE estimates for both RCC and FCC designs are centered on the same value, although the RCC has a tighter distribution.}

\begin{figure}
    \centering
    \includegraphics[width = \textwidth]{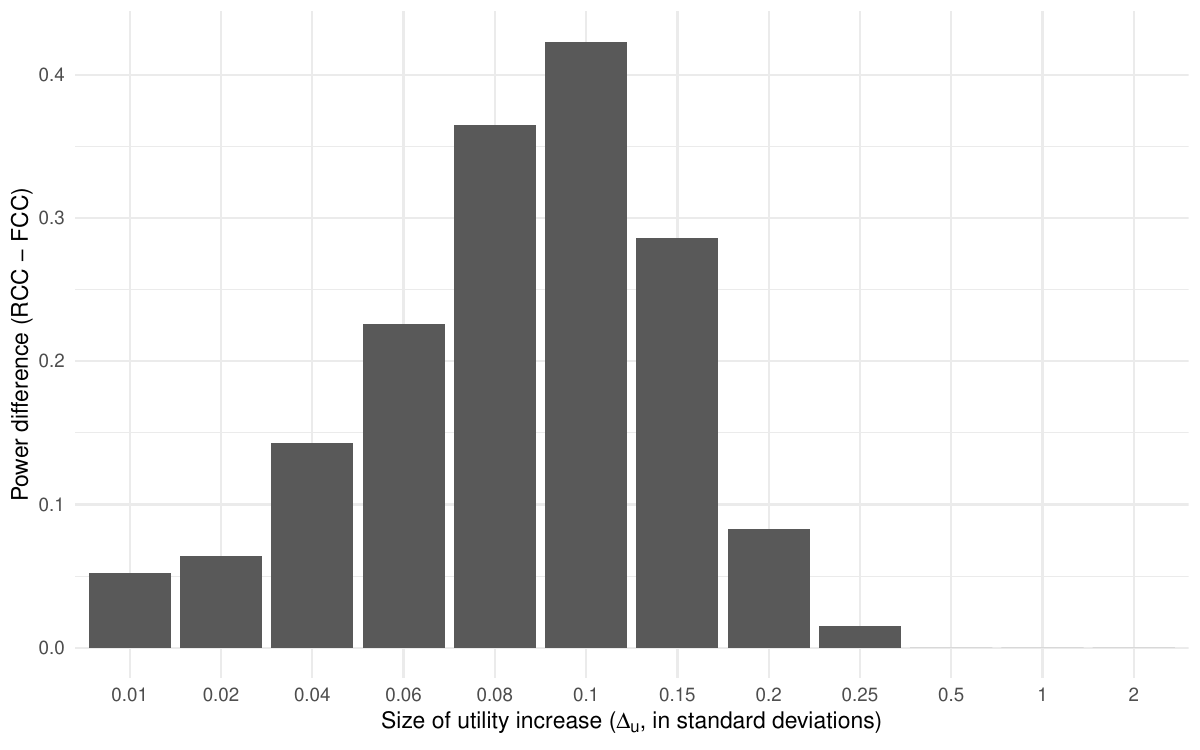}
    \caption{Estimated difference in power between RCC and FCC designs, holding constant the number of subjects and rounds.}
    \label{fig:power_difference}
\end{figure}

\section{Empirical comparison of forced versus ranked-choice designs}

We assess the ranked-choice conjoint design across two pre-registered survey experiments. Study 1 establishes the baseline utility of the ranked-choice conjoint model across two political science contexts--candidate elections and policy preferences--comparing it against the standard forced-choice design. Study 2 focuses on how increasing the number of profiles ($K = 2, 4, 6$) affects the design, and introduces tests for violations of the transitivity and IIA assumptions. We pre-registered seven hypotheses covering survey duration, satisficing, estimation efficiency, attribute importance, estimate stability, transitivity, and IIA (see Appendix Section \ref{sec:pre-reg} for the full list).

\subsection{Experimental design}

\paragraph{Study 1} In both experiments, subjects are randomized into completing either forced-choice or ranked-choice conjoint tasks. The first experiment is a \textit{candidate election} context, replicating \citet{Kirkland18}: subjects evaluate candidates for the US House described by age, race, gender, career, political experience, and party (Table~\ref{tab:attributes}). The second experiment concerns \textit{federal budget proposals}, with spending on five areas varying between ``decrease'', ``remain the same'', and ``increase'', inspired by \citet{bansak2021} (Table~\ref{tab:attributes2}).

\begin{table}[H]
	\caption{Conjoint attributes and levels: Candidate experiment.} \label{tab:attributes}
\begin{tabular}{llllll}
\toprule
Race     & Political experience       & Career experience    & Gender & Age & Party       \\
\midrule
White    & None                       & Educator             & Female & 35  & Independent \\
Hispanic & School board president     & Stay-at-home Mom/Dad & Male   & 45  & Democrat    \\
Black    & City council member        & Small business owner &        & 55  & Republican  \\
Asian    & State legislator           & Police officer       &        & 65  &             \\
         & Representative in congress & Electrician          &        &     &             \\
         & Mayor                      & Business executive   &        &     &             \\
         &                            & Attorney             &        &     &   \\
\bottomrule
\end{tabular}
\end{table}

\begin{table}[!htbp]
\caption{Conjoint attributes and levels: Budget experiment.} \label{tab:attributes2}
\resizebox{\textwidth}{!}{
\begin{tabular}{@{}ccccc@{}}
\toprule
Education spending & Environmental protection spending & Health and Medicare spending & Military spending & Social security spending \\ \midrule
Decrease & Decrease & Decrease & Decrease & Decrease \\
Remain the same & Remain the same & Remain the same & Remain the same & Remain the same \\
Increase & Increase & Increase & Increase & Increase \\\bottomrule
\end{tabular}}
\end{table}

All subjects completed both experiments in randomized order. The forced-choice arm completed $J = 6$ tasks with $K = 2$ profiles; the ranked-choice arm completed $J = 4$ tasks with $K = 3$ profiles (both evaluating 12 profiles per experiment). In the ranked-choice arm, subjects drag profile labels into descending preference order (Figure \ref{fig:vignette}). After each experiment, subjects complete factual manipulation checks \citep{Kane2019} and evaluate a single fixed profile. The survey was fielded in June 2024 on Prolific ($N = 2{,}211$).\footnote{Appendix Tables~\ref{tab:balance1} and~\ref{tab:balance2} confirm covariate balance across conditions in both studies.}

\begin{figure}
    \centering
    \includegraphics[scale = 0.35]{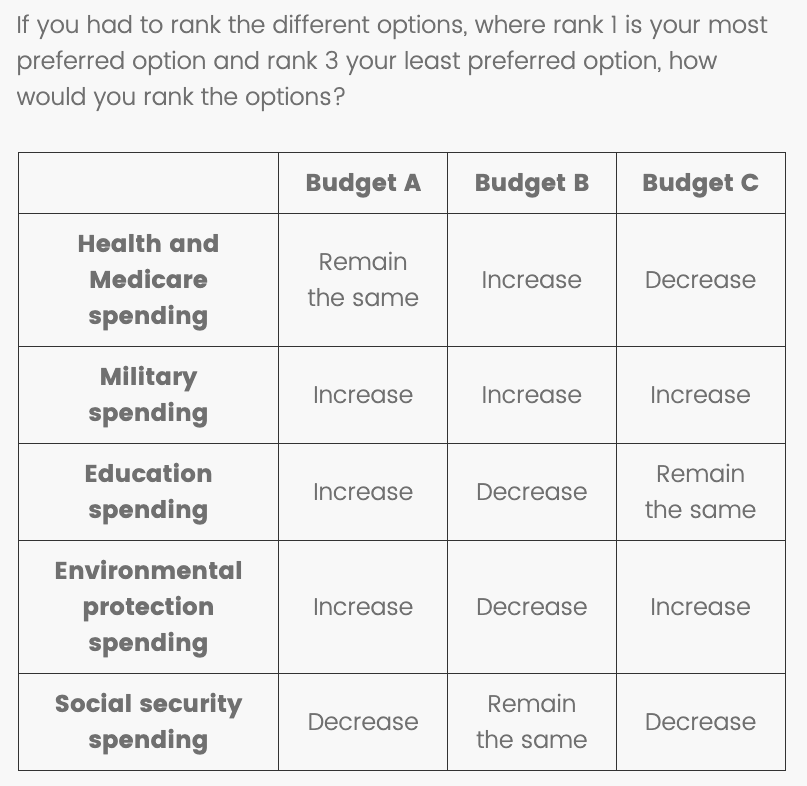}
    \caption{Example of the ranked-choice budget vignette.}
    \label{fig:vignette}
\end{figure}

\paragraph{Study 2} Our follow-up study replicates the candidate conjoint from Study 1, but randomizes subjects into one of three conditions: $K = 2$ (forced-choice), $K = 4$, or $K = 6$ profiles per vignette, each with $J = 4$ rounds. This design allows us to assess how AMCE estimates and assumption violations change as the number of profiles increases. To test transitivity (H6), after the main experiment we re-present two randomly selected profiles from the conjoint and ask subjects to choose between them; we then compare this forced choice against the implied ranking from the main experiment. To test IIA (H7), we repeat this exercise with two other profiles, but add a newly generated third profile and ask subjects to rank all three; we then check whether the relative ordering of the original two profiles is preserved. The survey was fielded in April 2025 on Prolific ($N = 2{,}837$; approximately 950 per condition).\footnote{The \texttt{cj\_test\_consistency()} function in the \textbf{cjrank} package implements these tests, computing violation rates by condition and conducting pairwise proportion tests between conditions.}

\section{Results}

\subsection{Survey duration and satisficing}

In Study 1, subjects completing ranked-choice tasks spent, on average, 44 seconds longer on the candidate experiment and over two minutes longer on the policy task (Table~\ref{tab:timings}). Critically, however, ranked-choice subjects completed 33\% fewer rounds (4 versus 6) for the same number of total profiles. Per effective observation gained, the ranked-choice design is approximately 3 seconds quicker per observation in the candidate experiment. Study 2 confirms this pattern at larger $K$: total time increases with the number of profiles, but the time per observation decreases substantially (Appendix Table~\ref{tab:timings_followup}). To assess whether the precision gains justify the additional time, we ask: how much precision does each design deliver per minute of survey time? Combining the squared standard errors with mean completion times, ranked-choice designs deliver 32\% ($K = 4$) to 42\% ($K = 6$) more precision per unit of survey time than the forced-choice baseline.

No attrition was observed in either study. To check for satisficing, we employed FMCs \citep{Kane2019}. In Study 1, ranked-choice subjects perform marginally worse on identifying the correct vignette format (by about 4 percentage points in both experiments), but show no difference on other measures (Table~\ref{tab:fmc}). Study 2 shows a similar pattern: subjects ranking 6 profiles are slightly worse at vignette identification, but no worse on other FMCs (Appendix Table~\ref{tab:fmc_followup}). We also find no evidence of systematic learning or fatigue effects across rounds: of 120 pairwise Z-tests comparing first-round and last-round AMCE estimates across conditions and studies, only 7 are significant at $p < 0.05$--consistent with chance (Appendix Figures~\ref{fig:amce_round_s1} and~\ref{fig:amce_round_s2}).

\begin{table}
\begin{center}
\begin{tabular}{l c c c c c c}
\toprule
 & \multicolumn{3}{c}{Total time taken} & \multicolumn{3}{c}{Time per observation} \\
\cmidrule(lr){2-4} \cmidrule(lr){5-7}
 & Candidate & Policy & Combined & Candidate & Policy & Combined \\
\midrule
Ranked-choice   & $44.40^{***}$  & $134.26^{**}$  & $44.40$        & $-2.72^{***}$ & $-0.04$       & $-2.72$      \\
                & $(7.02)$       & $(50.39)$      & $(35.97)$      & $(0.35)$      & $(2.11)$      & $(1.51)$     \\
Policy          &                &                & $25.67$        &               &               & $2.14$       \\
                &                &                & $(35.70)$      &               &               & $(1.50)$     \\
Ranked x Policy &                &                & $89.86$        &               &               & $2.67$       \\
                &                &                & $(50.87)$      &               &               & $(2.14)$     \\
(Intercept)     & $109.67^{***}$ & $135.34^{***}$ & $109.67^{***}$ & $9.14^{***}$  & $11.28^{***}$ & $9.14^{***}$ \\
                & $(4.93)$       & $(35.36)$      & $(25.25)$      & $(0.25)$      & $(1.48)$      & $(1.06)$     \\
\midrule
R$^2$           & $0.02$         & $0.00$         & $0.01$         & $0.03$        & $0.00$        & $0.00$       \\
Adj. R$^2$      & $0.02$         & $0.00$         & $0.00$         & $0.03$        & $-0.00$       & $0.00$       \\
Num. obs.       & $2211$         & $2211$         & $4422$         & $2211$        & $2211$        & $4422$       \\
\bottomrule
\multicolumn{7}{l}{\scriptsize{$^{***}p<0.001$; $^{**}p<0.01$; $^{*}p<0.05$}}
\end{tabular}
\caption{Effect of ranked choice (compared to forced choice) on completion time}
\label{tab:timings}
\end{center}
\end{table}

\begin{table}
\begin{center}
\begin{tabular}{l c c c c c c}
\toprule
 & \multicolumn{2}{c}{Number of attributes} & \multicolumn{2}{c}{Identified vignette} & \multicolumn{2}{c}{Dimensions} \\
\cmidrule(lr){2-3} \cmidrule(lr){4-5} \cmidrule(lr){6-7}
 & Candidate & Policy & Candidate & Policy & Candidate & Policy \\
\midrule
Ranked-choice & $0.28^{***}$ & $-0.02$      & $-0.04^{*}$  & $-0.03^{**}$ & $-0.03$      & $-0.01$      \\
              & $(0.03)$     & $(0.03)$     & $(0.02)$     & $(0.01)$     & $(0.03)$     & $(0.01)$     \\
(Intercept)   & $0.65^{***}$ & $0.52^{***}$ & $0.78^{***}$ & $0.97^{***}$ & $2.56^{***}$ & $0.98^{***}$ \\
              & $(0.02)$     & $(0.02)$     & $(0.01)$     & $(0.01)$     & $(0.02)$     & $(0.01)$     \\
\midrule
R$^2$         & $0.03$       & $0.00$       & $0.00$       & $0.00$       & $0.00$       & $0.00$       \\
Adj. R$^2$    & $0.03$       & $-0.00$      & $0.00$       & $0.00$       & $-0.00$      & $0.00$       \\
Num. obs.     & $2160$       & $2197$       & $2209$       & $2209$       & $2211$       & $2210$       \\
\bottomrule
\multicolumn{7}{l}{\scriptsize{$^{***}p<0.001$; $^{**}p<0.01$; $^{*}p<0.05$}}
\end{tabular}
\caption{Effect of ranked choice (compared to forced choice) on manipulation checks}
\label{tab:fmc}
\end{center}
\end{table}

\subsection{AMCE estimates and efficiency}

\paragraph{Study 1} Figures \ref{fig:mod_cand} and \ref{fig:mod_policy} plot the estimated AMCEs from the candidate and policy experiments. The substantive interpretation of the AMCEs is very similar across designs: in almost all cases, coefficients share the same sign and order of magnitude. What is notable, however, is that confidence intervals are narrower in every case for the ranked-choice condition. On average, the RCC design yields standard errors that are 12\% smaller than the forced-choice equivalent in the candidate experiment and 13\% smaller in the policy experiment, using standard errors clustered at the subject level throughout.\footnote{These empirical gains are attenuated relative to the theoretical prediction of 33\% (Proposition~\ref{prop:variance}), reflecting treatment effect heterogeneity and increased intra-cluster correlation. Without clustering, the SE reduction is approximately 25\%, illustrating the importance of accounting for the within-subject dependence created by rank expansion.} This efficiency gain is achieved while requiring subjects to evaluate 33\% fewer rounds (4 versus 6), reducing the repetitive burden of the conjoint task.

One implementation consideration is profile display position. We find a primacy bias in the drag-and-drop ranking interface--profiles listed lower receive systematically worse ranks--but this does not affect AMCE estimates since attributes are balanced across positions by randomization. We recommend randomizing display order and, as a robustness check, including position fixed effects (see Appendix Table~\ref{tab:position_effects}).

\begin{figure}[tbp]
    \centering
    \includegraphics[width = \textwidth]{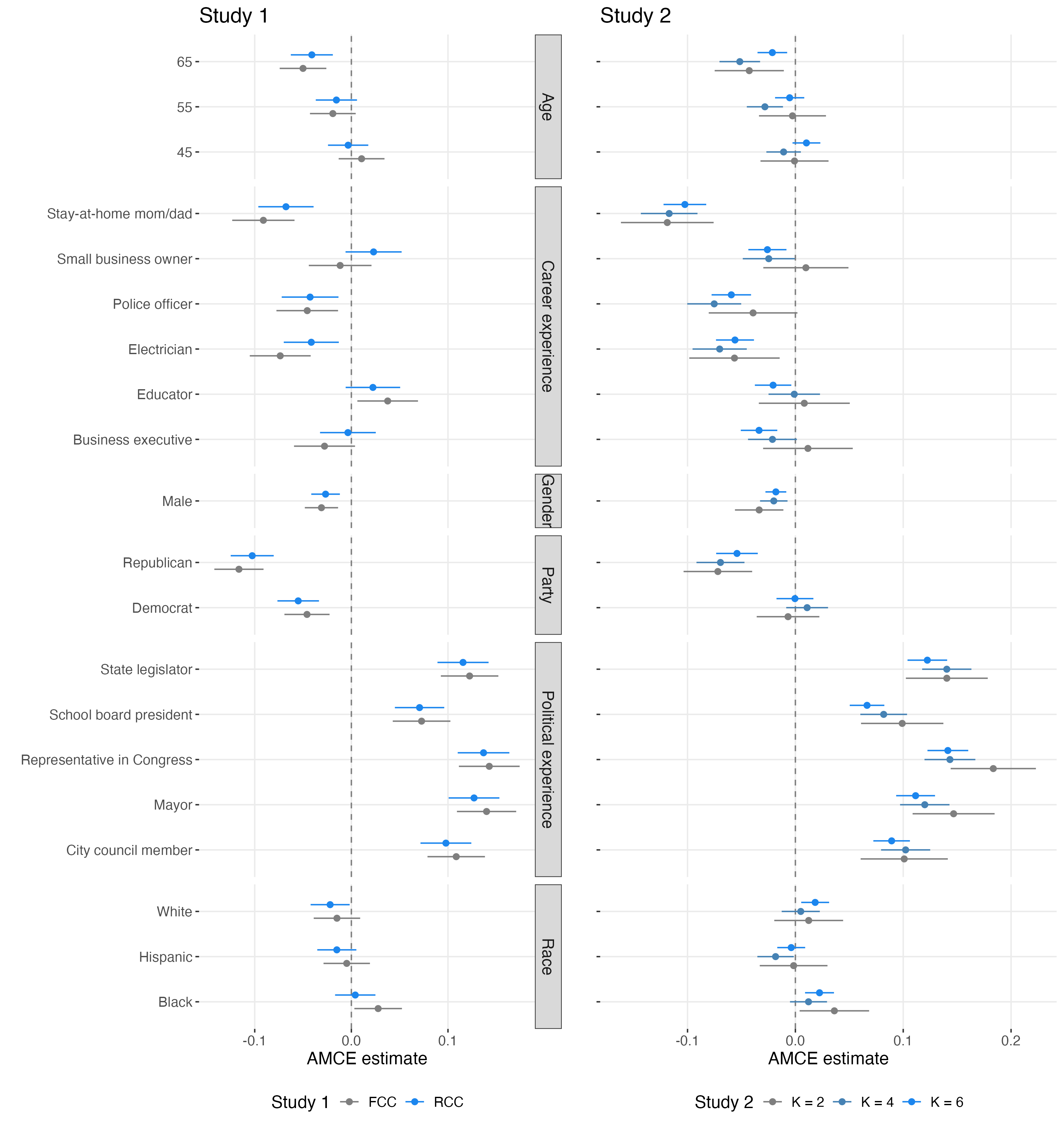}
    \caption{Estimated AMCEs from the candidate experiment. Left panel: Study 1 (FCC vs RCC, $K = 2$ vs $K = 3$). Right panel: Study 2 ($K = 2, 4, 6$). Point estimates are stable across designs; confidence intervals narrow with increasing $K$.}
    \label{fig:mod_cand}
\end{figure}

The AMCEs from the policy experiment exhibit a similar pattern of consistency across conditions, with ranked-choice estimates closely tracking their forced-choice counterparts (Appendix Figure \ref{fig:mod_policy}).

\paragraph{Study 2} Figure \ref{fig:mod_cand} (right panel) shows the AMCE estimates from the candidate experiment across $K = 2, 4, 6$ profiles. The key finding is that the substantive pattern of AMCEs is remarkably stable across conditions: coefficients are similar in sign, magnitude, and significance, despite the very different amounts of information in each vignette. Standard errors decrease substantially with $K$: relative to the $K=2$ baseline, mean clustered SEs are 42\% smaller at $K=4$ and 55\% smaller at $K=6$. These are attenuated relative to the theoretical predictions from Proposition~\ref{prop:variance} (47\% and 61\% respectively), consistent with the increasing intra-cluster correlation as the rank expansion generates more dependent rows per subject.

Contrary to H5, we do not find evidence that increasing the number of profiles introduces substantial noise into the AMCE estimates. To assess this formally, we follow our pre-registered Z-tests comparing each coefficient across condition pairs: of 20 coefficients, none differ significantly between $K = 2$ and $K = 4$, and only 2 differ between $K = 2$ and $K = 6$ (at $p < 0.05$). This stability of point estimates suggests that, at least up to $K = 6$, the additional cognitive burden of ranking more profiles does not systematically bias the marginal effect estimates. Moreover, the efficiency gain is not merely an artifact of showing respondents more profiles: comparing models fit on the \textit{same number of profiles per respondent}--1 round of $K = 6$ versus 3 rounds of $K = 2$ (both using 6 profiles)--the ranked-choice design yields standard errors that are 30\% smaller. This gain comes from how many pairwise comparisons each profile participates in (every ranked profile enters $K-1$ comparisons, each forced-choice profile enters 1), rather than from middle-ranked pairs carrying qualitatively different preference signal. The structural analysis in Appendix Table~\ref{tab:sconjoint} confirms that holding the number of pairs constant, ranking and forced-choice pairs are equivalent in kind.

\subsection{Assumption validity}
\label{sec:assumptions}

Study 2 allows us to directly test the transitivity and IIA assumptions that underpin the rank expansion. Figure \ref{fig:violations} plots the violation rates by condition. For transitivity, the baseline violation rate in the $K=2$ condition--which reflects general measurement noise from re-presenting a forced-choice comparison--is approximately 13\%. This increases significantly to around 22\% for $K = 4$ and $K = 6$ ($p < 0.001$ for both comparisons against $K = 2$), but does not differ between $K = 4$ and $K = 6$ ($p = 1$).

For IIA violations, the baseline rate is higher at approximately 19\% for $K = 2$, rising to 22\% for $K = 4$ ($p = 0.064$) and 25\% for $K = 6$ ($p < 0.001$ versus $K = 2$). Again, the difference between $K = 4$ and $K = 6$ is not statistically significant ($p = 0.164$).

\begin{figure}
    \centering
    \includegraphics[width = 0.48\textwidth]{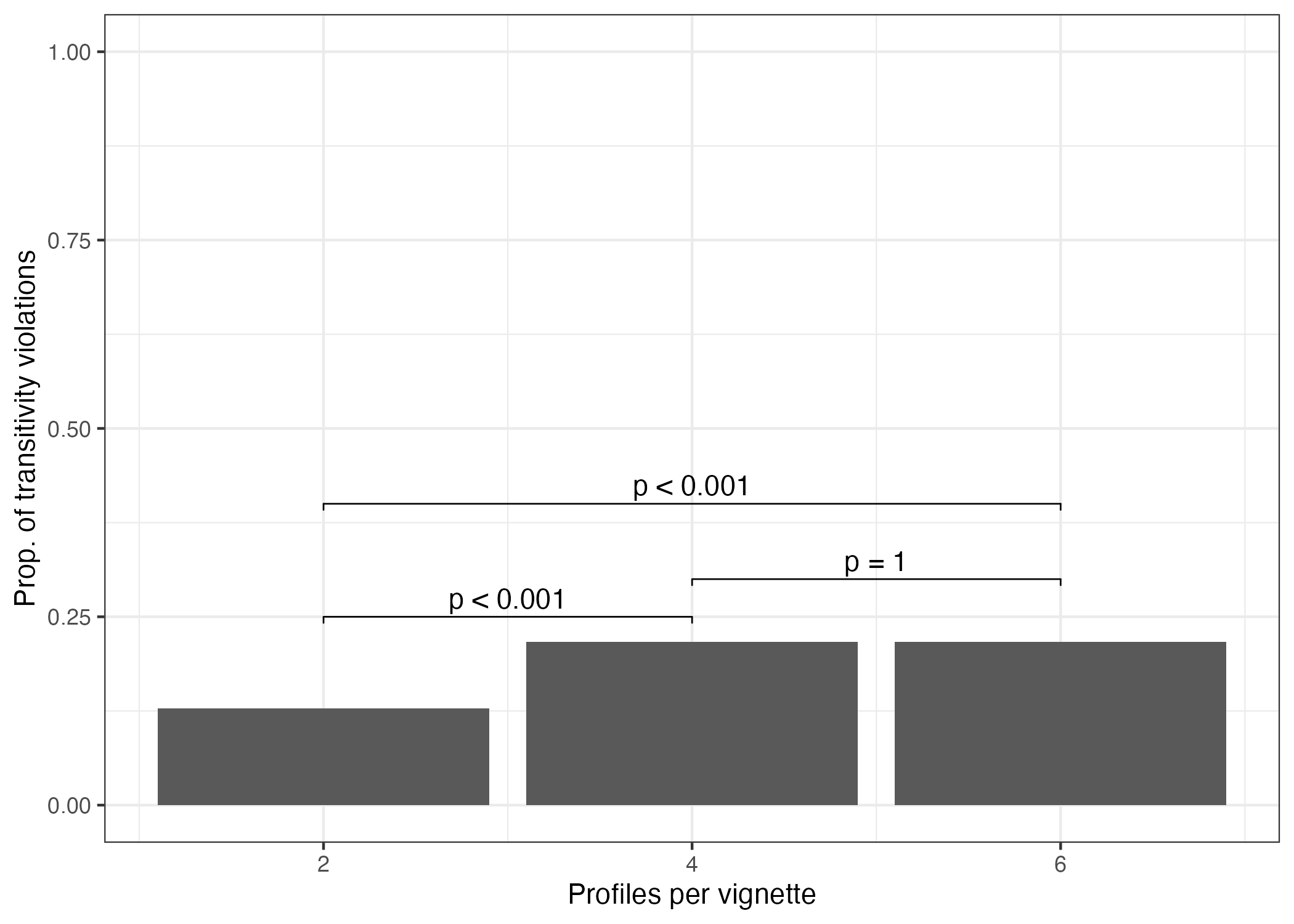}
    \includegraphics[width = 0.48\textwidth]{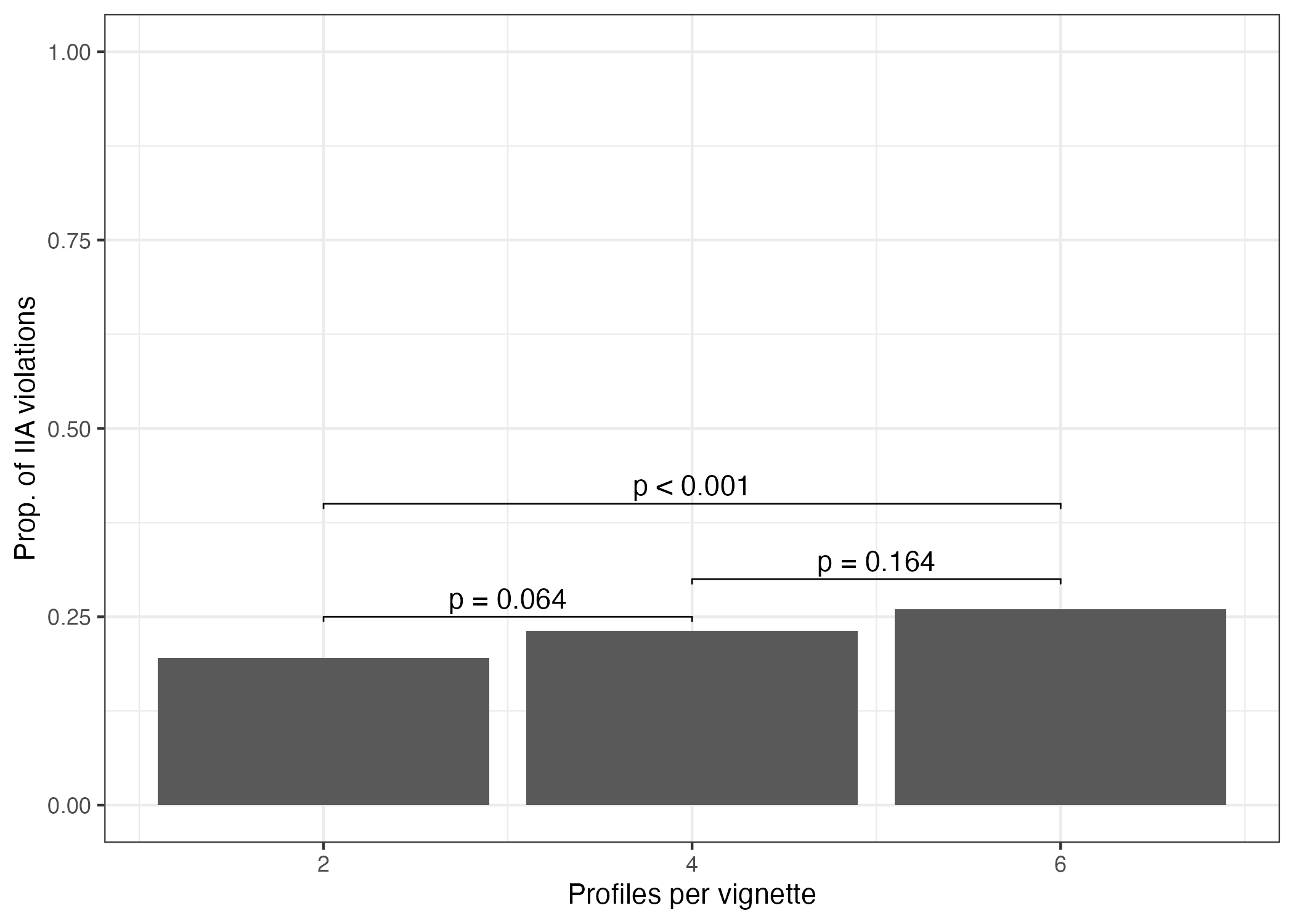}
    \caption{Study 2: Proportion of transitivity (left) and IIA (right) violations by number of profiles per vignette. Brackets show $p$-values from proportion tests.}
    \label{fig:violations}
\end{figure}

Two findings are noteworthy. First, even the forced-choice baseline ($K=2$) exhibits non-trivial violation rates, suggesting that a meaningful share of measurement noise is intrinsic to the conjoint setup rather than specific to ranking. Reframed as test-retest reliability, the $K=2$ rates imply that standard forced-choice conjoint has approximately 87\% (transitivity) and 81\% (IIA) reliability when the same comparison is re-presented. This is broadly consistent with the intra-respondent reliability documented by \citet{clayton2023correcting} across eight replicated conjoint studies, and provides a useful field-specific benchmark. \citeauthor{clayton2023correcting} call for further measurement error research on non-binary conjoint designs including rankings; our results suggest that the additional measurement error introduced by ranking is modest relative to the baseline noise already present in forced-choice designs.

Second, while violations increase with $K$, they do so modestly and plateau between $K = 4$ and $K = 6$. Importantly, re-estimating the AMCEs after excluding subjects who exhibited either violation yields nearly identical coefficient estimates (Appendix Figure \ref{fig:mod_noviolations}), indicating that these violations do not materially bias the AMCE.

We examine individual-level predictors of violations in Appendix Table~\ref{tab:violation_predictors}; notably, completion time does not predict violations, indicating they are not simply a consequence of rushing. A simulation exercise (Appendix Figure \ref{fig:sensitivity}) confirms that at the empirically observed violation rate of approximately 22\%, the mean absolute deviation of AMCEs from the uncorrupted baseline is small (approximately 0.026 percentage points), reinforcing the robustness of the estimates.

\subsection{Recovering preferences}

A central question is whether the rank outcome captures richer preference information than the binary forced-choice outcome. We assess this through intensity ratings and out-of-sample prediction.

\paragraph{Intensity ratings} In Study 2, subjects rated each profile on a 0-100 scale indicating how likely they would vote for that candidate--a continuous measure of preference independent of the ranking task.\footnote{This analysis was pre-registered as exploratory. See the Study 2 pre-analysis plan, exploratory analysis \#2.} At the aggregate level, the AMCEs estimated from choice/rank data correlate strongly with those estimated from the intensity ratings across all conditions: $r = 0.92 \ (\text{SE} = 0.091)$ for $K = 2$, $r = 0.96 \ (\text{SE} = 0.068)$ for $K = 4$, and $r = 0.98 \ (\text{SE} = 0.044)$ for $K = 6$ (Appendix Figure \ref{fig:int_vs_amce}). Neither design's AMCE-derived importance scores are consistently closer to self-reported attribute rankings (see Appendix Section~\ref{sec:pre-reg}).

At the individual level, the within-subject correlation between the normalized rank outcome and the intensity rating increases from 0.47 ($K = 2, \ \text{SE} = 0.010$) to 0.60 ($K = 4, \ \text{SE} = 0.007$) to 0.65 ($K = 6, \ \text{SE} = 0.006$). This reflects the core measurement advantage of ranking: the forced-choice outcome is necessarily binary, whereas the normalized rank has $K$ levels, each of which captures genuine gradations in preference intensity. To verify that this is not merely a statistical artifact, we computed the correlation between a binarized top-choice indicator and the intensity rating across conditions. This correlation is stable (0.47 [SE $= 0.010$], 0.46 [SE $= 0.007$], 0.43 [SE $= 0.006$]) confirming that ranking does not \textit{degrade} the quality of the most important comparison--while the intermediate ranks carry genuine preference content that the binary outcome does not collect within a single vignette. Figure~\ref{fig:rank_intensity} visualizes this pattern across conditions. This is the measurement-level intuition behind Proposition \ref{prop:rankequiv}: the normalized rank $\tilde{R}$ averages over $K-1$ pairwise comparisons per profile, and the intensity data validates that each of those additional comparisons carries real preference information.

\paragraph{Out-of-sample prediction} We also train gradient-boosted tree models on each condition's conjoint data, including pre-treatment covariates, and evaluate out-of-sample prediction on post-treatment profiles.\footnote{Models are tuned via Optuna's TPE sampler with 10-fold cross-validation and 1000 hyperparameter samples. We use the area under the ROC curve (AUC) as our classification metric and DeLong's test for comparing AUC scores across conditions.} We conduct two types of this out-of-sample test.

First, across both studies, we show one fixed profile to all subjects, which was constructed to be hard to decide over (for example, the candidate is described as an Independent, and so strong partisan cues are unavailable). Figure \ref{fig:oos} plots each experimental condition's predictive performance over this profile, varying the number of profiles. In Study 1 (Table~\ref{tab:roc}), the ranked-choice model substantially outperforms the forced-choice model in the candidate experiment. In the policy experiment, the ranked-choice model does not outperform forced choice, consistent with the weaker efficiency gains observed in that domain. In Study 2, $K = 4$ yields the highest AUC on fixed test profiles in the same candidate experiment (0.627 versus 0.448 for $K = 2$; $p < 0.001$).\footnote{We note that these out-of-sample predictions are made against a binary test outcome (vote yes/no for a fixed profile), which is itself subject to the swapping error documented by \citet{clayton2023correcting}.} The modest AUC levels across all conditions (0.5-0.7) are consistent with substantial measurement noise in both the training and test data; our key finding is the relative comparison across conditions rather than the absolute prediction accuracy.\footnote{The AUC patterns are not model-dependent: logistic regression yields qualitatively similar results across conditions (Appendix Table~\ref{tab:roc_logistic}).}

\begin{figure}
    \centering
    \includegraphics[width=0.95\linewidth]{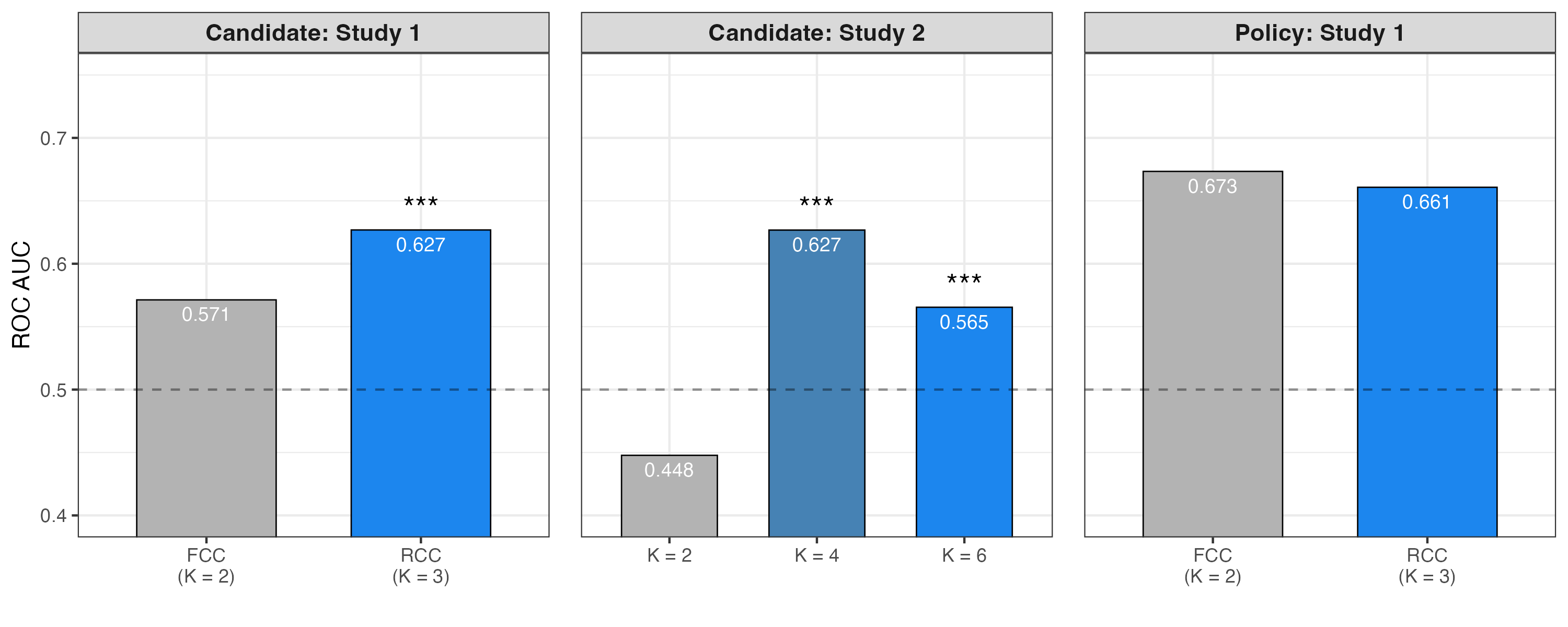}
    \caption{Out-of-sample classification performance (ROC-AUC) across both studies. The first two panels indicate the performance in the candidate conjoint context; the last panel shows the results in the policy conjoint context.}
    \label{fig:oos}
\end{figure}

In our second study, we also presented another post-treatment test profile but where the values of all the attributes were randomized--allowing us to assess the predictive performance over the span of the attributes in the candidate study. In average terms, we find the ranked-choice $K=4$ model performs significantly better than the forced-choice design, although the absolute size of this difference is small (see Appendix Table \ref{tab:roc_followup}). More interestingly, as shown in Figure \ref{fig:oos_reg}, ranking yields the largest predictive gains for profiles that are hardest to classify from forced-choice data alone.

\begin{figure}
    \centering
    \includegraphics[width=0.8\linewidth]{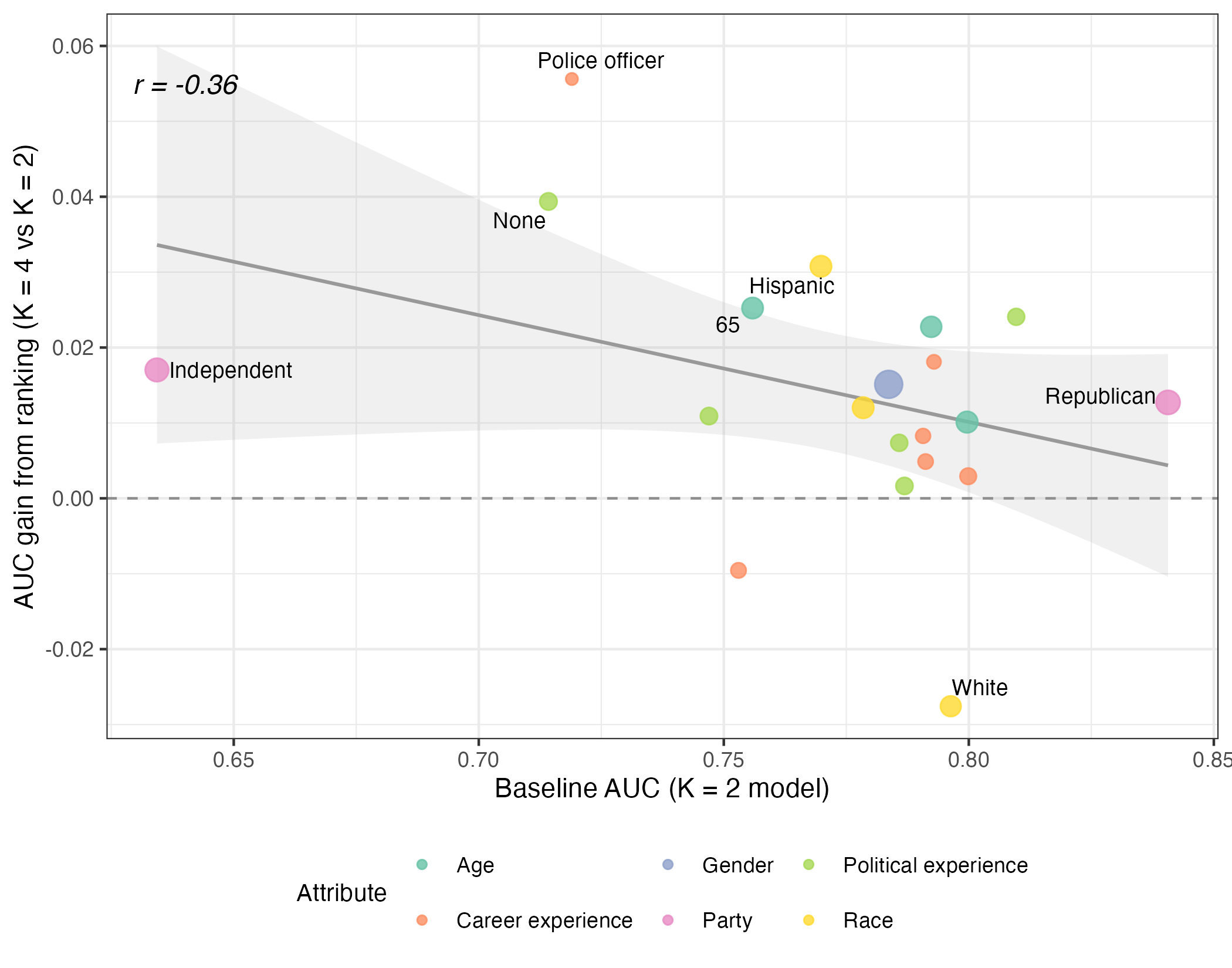}
     \caption{AUC gain from ranking ($K = 4$ versus $K = 2$) as a function of baseline prediction difficulty, by attribute-level. Each point represents a subset of the Study 2 randomized test profiles sharing a given attribute level (e.g.\ Party = Independent, Race = Hispanic); point size is proportional to subset size. The $x$-axis shows the AUC achieved by the forced-choice ($K = 2$) model on that subset, and the $y$-axis shows the improvement when using the ranked-choice ($K = 4$) model instead.}
    \label{fig:oos_reg}
\end{figure}

As a robustness check using a fundamentally different estimation strategy, we also fit the structural deep-learning random-utility model of \citet{acharya2026learningpreferencesconjointdata} to the rank-expanded pairs in Study 2, which confirms that the predictive gain reflects pair quantity rather than a qualitatively different kind of preference signal (see Appendix \ref{sec:sconjoint} for details).

\section{Discussion}

In this paper, we introduce how ranked choices can be incorporated into conjoint experiments. This work builds on recent contributions in political science that point to the direct relevance, and potential performance benefits, of using rankings in the study of political phenomena \citep[e.g.,][]{kaufman2021measure, atsusaka2022causal,atsusaka_2023}. Under a faithful-ranking assumption and an expansion of ranks back to implicit pairwise comparisons, we preserve the design mechanics of the traditional conjoint experiment while collecting substantially more information from the same number of randomized profiles. We formalize this advantage in Proposition~\ref{prop:variance}, which shows that the theoretical SE reduction grows rapidly with the number of profiles: from 33\% at $K=3$ to 61\% at $K=6$.

Across both studies, the AMCEs estimated from ranked-choice data are reassuringly consistent with their forced-choice counterparts, and are more precisely estimated. The practical value of this efficiency depends on the type of constraints a researcher faces.  When \textit{sample size} is the bottleneck--as is increasingly the case with probability-based panels where per-respondent costs are high--the precision gains translate directly into sample-size savings. The cost-effectiveness of ranked-choice designs is especially attractive given rising concerns about the quality and cost of convenience samples for survey experiments \citep{freese2025online}. Ranked-choice conjoints can deliver equivalent precision with fewer respondents, or greater precision with the same sample size--both of which reduce per-respondent costs and make higher-quality probability samples more feasible for conjoint research.\footnote{We note our evidence comes from Prolific convenience samples; whether violation rates and time costs transfer to probability panels is an open question.}

When \textit{survey time} is the binding constraint, ranking tasks may take longer per round, but subjects complete fewer rounds for the same number of total profiles (33\% fewer in Study 1). The trade-off between time and precision works in ranked conjoints' favor: ranked-choice designs deliver 32-42\% more precision per minute of survey time at $K = 4$ and $K = 6$ compared to $K = 2$. Even when comparing designs that evaluate the \textit{same number of total profiles}, ranking yields 30\% smaller standard errors, confirming that the efficiency gain is not merely a consequence of more observations but reflects how many pairwise comparisons each vignette yields. By contrast, achieving similar gains by adding more forced-choice rounds would increase respondent fatigue and exposure to carryover effects \citep{Ham_Imai_Janson_2024}.

Our intensity rating analysis validates the measurement component of this efficiency story directly: AMCEs from choice/rank data correlate strongly ($r > 0.92, \ [\text{SE} < 0.10]$) with those from independent 0-100 intensity ratings, and the binarized top-choice signal is equally strong regardless of $K$, confirming that ranking does not degrade the most important comparison. The additional pairwise comparisons that ranking collects per vignette carry genuine preference content, rather than noise--and a structural random-utility estimator (Appendix Table~\ref{tab:sconjoint}) confirms that these extra comparisons behave as equivalent in kind to forced-choice pairs, but are more numerous. An alternative would be to collect continuous ratings directly, avoiding the transitivity and IIA assumptions. However, rankings preserve the choice-based framework that AMCEs are designed for. Rating scales may also introduce scale-use heterogeneity and anchoring effects. Empirically, AMCEs from intensity ratings closely track those from rank-expanded data (Appendix Figure \ref{fig:amce_intensity}); we view the two approaches as complementary.

A key concern with ranking more profiles is whether respondents can faithfully rank them. Our bespoke tests of transitivity and IIA show that violation rates increase modestly with $K$ (Section \ref{sec:assumptions}), but the key question is whether these violations distort the estimates. Panel (B) of Figure \ref{fig:cost_benefit} shows they do not: the mean absolute AMCE deviation when excluding all violators is negligible at every value of $K$, confirming that the rank-expanded estimates are robust to the empirically observed levels of transitivity and IIA violations. The baseline forced-choice condition ($K = 2$) itself exhibits approximately 13\% transitivity violation rates when comparisons are re-presented, establishing a useful benchmark for the test-retest reliability of conjoint experiments more generally. Together, these patterns suggest that $K = 4$ offers the most attractive balance of precision, validity, and respondent burden for most applications.

Figure \ref{fig:cost_benefit} summarizes the key trade-offs across values of $K$, including the sample-size multipliers implied by the empirical SE reductions.

\begin{figure}[tbp]
    \centering
    \includegraphics[width = \textwidth]{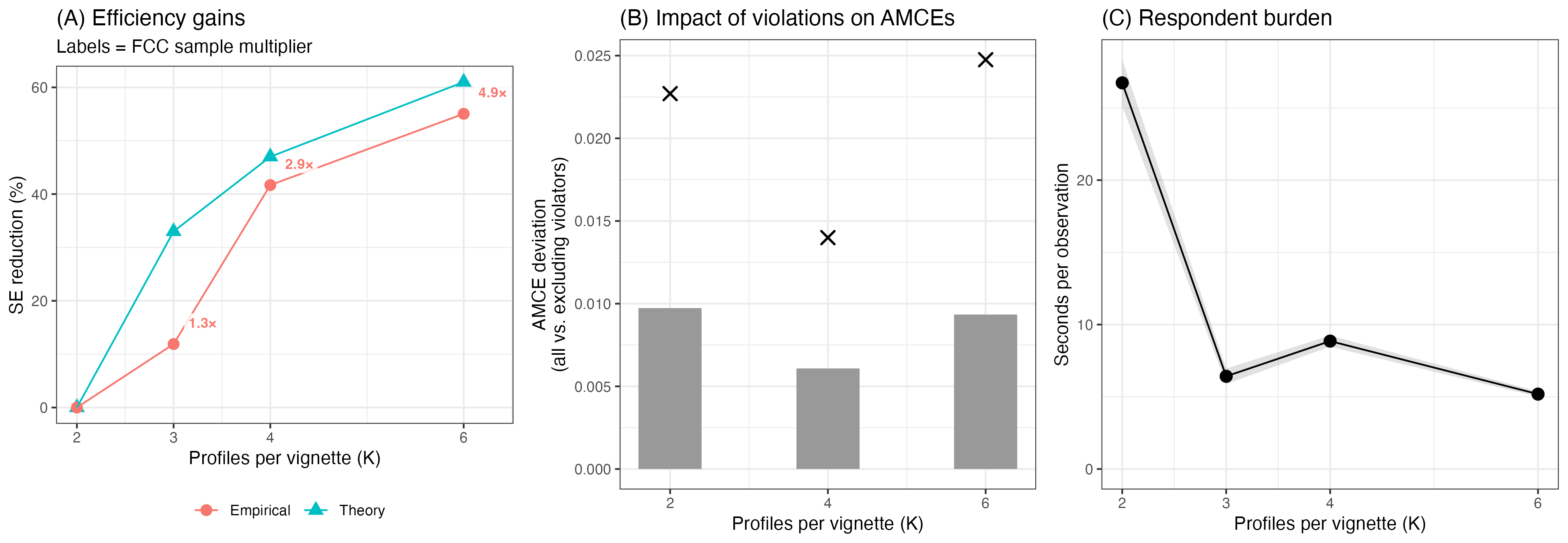}
    \caption{Summary of trade-offs across values of $K$. Panel (A) shows theoretical (Proposition~\ref{prop:variance}) and empirical SE reductions relative to $K = 2$; bold labels indicate the FCC sample-size multiplier (number of forced-choice respondents needed per ranked-choice respondent to achieve equivalent precision). Panel (B) shows the mean absolute deviation in AMCE estimates when subjects who violated transitivity or IIA are excluded, with crosses marking the maximum single-coefficient deviation; deviations are negligible at all values of $K$. Panel (C) shows mean time per effective observation. Data at $K = 3$ are from Study 1; all other values from Study 2.}
    \label{fig:cost_benefit}
\end{figure}

An informative feature of our results is the difference in performance across political contexts, which offers practical guidance for implementation. In the candidate experiment, we see clear efficiency and predictive advantages for ranking--unsurprisingly, given the direct analogue to real electoral systems where voters routinely compare and rank candidates. The budget policy experiment, by contrast, yields efficiency gains but weaker predictive improvements; subjects were also slower, suggesting that the symmetric, abstract attribute structure of budget items makes the ranking task more demanding. We view this variation as a useful guide for researchers rather than a limitation: ranked-choice designs are likely to perform best when profiles represent familiar objects that subjects naturally evaluate comparatively (candidates, consumer products, housing options) and may offer less benefit when profiles are defined by symmetric, interchangeable dimensions. We recommend that researchers considering ranked-choice designs assess the naturalness of the ranking task in their specific domain, and pilot with a small sample to verify that the additional cognitive demand does not introduce excessive noise.

To support researchers in implementing ranked-choice conjoint designs, we provide the \textbf{cjrank} \texttt{R} package \citep[available at][]{cjrank2026}. The package provides an end-to-end workflow, including diagnostic tests of the critical assumptions our method relies on. The package is designed to integrate with existing conjoint analysis workflows, requiring only that researchers supply long-format data with a rank column.

\newpage







\bibliographystyle{agsm}
\bibliography{cjoint}

\appendix
\newpage
\begin{center}
    \section*{\huge Online Appendices for \\ Ranked-choice conjoint experiments}
    \large Thomas S. Robinson, Mats Ahrenshop, and Spyros Kosmidis
\end{center}
\pagenumbering{roman}   

\addcontentsline{toc}{section}{Appendices}
\renewcommand{\thesection}{A\arabic{section}}
\renewcommand{\thefigure}{A\arabic{figure}}
\renewcommand{\thetable}{A\arabic{table}}
\renewcommand{\theHfigure}{A\arabic{figure}}
\renewcommand{\theHtable}{A\arabic{table}}
\setcounter{figure}{0}
\setcounter{table}{0}
\setcounter{page}{1}

\section{Proof of unbiasedness of AMCE under rank expansion}
\label{sec:proof}

We prove that, under two additional assumptions, the AMCE estimated from rank-expanded data is identical to the AMCE that would be obtained from a forced-choice conjoint. Throughout, we adopt the notation introduced in the main text.

\subsection{Setup and notation}

Subject $i$ is shown $K$ profiles in a single round of a conjoint experiment. Each profile $p$ is characterized by an attribute vector $X_p = (X_{p1},\ldots,X_{pL})$, where $X_{pl}$ denotes the realized level of attribute $l$. Subject $i$ derives latent utility
$$
U_{ip} = \beta^\intercal X_p + \epsilon_{ip},
$$
where $\beta$ is a common preference vector and $\epsilon_{ip}$ is a mean-zero stochastic component satisfying $\mathbb{E}[\epsilon_{ip} \mid X_p] = 0$. We assume throughout that all attribute levels are independently and uniformly randomized across profiles, subjects, and rounds.

Let $\mathbf{P}_i = \{p_1,\ldots,p_K\}$ denote the set of profiles shown in a given round. In a \textit{forced-choice} design ($K=2$), the observed outcome is:
$$
Y^{\text{FC}}_{ipp'} = \mathbf{1}(U_{ip} > U_{ip'}).
$$
\citet{Hainmueller2014} show that, under independent randomization of attributes, OLS applied to forced-choice data identifies the AMCE. We take this result as given.

In a \textit{ranked-choice} design ($K \geq 2$), the subject provides a complete ranking $R_i : \mathbf{P}_i \to \{1,\ldots,K\}$, where $R_{ip} = 1$ denotes the most-preferred profile and $R_{ip} = K$ the least-preferred. Rank expansion constructs $2\binom{K}{2}$ pairwise observations from each round:
$$
\tilde{Y}_{ipp'} = \mathbf{1}(R_{ip} < R_{ip'}), \qquad \forall\; p \neq p' \in \mathbf{P}_i.
$$

\subsection{Assumptions}

\begin{assumption}[Faithful ranking]
\label{ass:transitive}
For each subject $i$ and round, the reported ranking is a faithful representation of the latent utility ordering. That is, for all $p, p' \in \mathbf{P}_i$:
$$
R_{ip} < R_{ip'} \iff U_{ip} > U_{ip'} \qquad \text{a.s.}
$$
\end{assumption}

This assumption requires that subjects rank profiles in accordance with their latent utilities. It implies, and is stronger than, the transitivity of preferences: if $U_{ip} > U_{ip'} > U_{ip''}$, then $R_{ip} < R_{ip'} < R_{ip''}$. Note that rankings are total orders by construction, so they cannot themselves be intransitive. The assumption can nevertheless fail if stochastic choice noise or inattentive responding causes a subject's reported ranking to deviate from their latent utility ordering--for example, if $U_{ip} > U_{ip'}$ but the subject reports $R_{ip} > R_{ip'}$. We discuss empirical tests of this assumption in the main text.

\begin{assumption}[Pairwise independence / IIA]
\label{ass:iia}
The pairwise comparison between any two profiles $p$ and $p'$ is independent of the attributes of all other profiles in the vignette:
$$
\Pr(R_{ip} < R_{ip'} \mid X_p, X_{p'}, X_{-(pp')}) = \Pr(R_{ip} < R_{ip'} \mid X_p, X_{p'}),
$$
where $X_{-(pp')}$ denotes the attribute vectors of the remaining $K-2$ profiles.
\end{assumption}

This is the independence of irrelevant alternatives (IIA) condition applied to rankings within a conjoint vignette. It requires that the relative ordering of any two profiles depends only on their own attributes, not on what other options are available. Section~\ref{sec:iia_violation} below discusses the consequences of violations.

\subsection{Main results}

\begin{lemma}[Pairwise equivalence]
\label{lem:pairwise}
Under Assumption~\ref{ass:transitive},
$$
\tilde{Y}_{ipp'} = Y^{\text{FC}}_{ipp'} \qquad \text{a.s.}
$$
That is, the pairwise outcome derived from the ranking equals the outcome that would obtain under a direct forced choice between the same two profiles.
\end{lemma}

\begin{proof}
By Assumption~\ref{ass:transitive}, for any $p, p' \in \mathbf{P}_i$:
$$
\tilde{Y}_{ipp'} = \mathbf{1}(R_{ip} < R_{ip'}) = \mathbf{1}(U_{ip} > U_{ip'}) = Y^{\text{FC}}_{ipp'} \qquad \text{a.s.}
$$
The first equality is the definition of $\tilde{Y}$. The second follows directly from Assumption~\ref{ass:transitive}. The third is the definition of $Y^{\text{FC}}$.
\end{proof}

\begin{proposition}[Unbiasedness of AMCE under rank expansion]
\label{prop:unbiased}
Under Assumptions~\ref{ass:transitive} and~\ref{ass:iia}, the OLS estimator applied to the rank-expanded dataset $\{(\tilde{Y}_{ipp'}, X_p, X_{p'})\}$ recovers the AMCE.
\end{proposition}

\begin{proof}
We show that the rank-expanded observations satisfy the same conditional moment restrictions as forced-choice data, so the identification result of \citet{Hainmueller2014} applies directly.

\textit{Step 1.} By Lemma~\ref{lem:pairwise}, $\tilde{Y}_{ipp'} = \mathbf{1}(U_{ip} > U_{ip'})$ a.s. Therefore:
$$
\mathbb{E}[\tilde{Y}_{ipp'} \mid X_p, X_{p'}, X_{-(pp')}] = \Pr(U_{ip} > U_{ip'} \mid X_p, X_{p'}, X_{-(pp')}).
$$

\textit{Step 2.} Since $\{R_{ip} < R_{ip'}\} = \{U_{ip} > U_{ip'}\}$ a.s.\ by Assumption~\ref{ass:transitive}, Assumption~\ref{ass:iia} implies that the right-hand side does not depend on $X_{-(pp')}$:
$$
\mathbb{E}[\tilde{Y}_{ipp'} \mid X_p, X_{p'}, X_{-(pp')}] = \Pr(U_{ip} > U_{ip'} \mid X_p, X_{p'}).
$$
Taking iterated expectations:
$$
\mathbb{E}[\tilde{Y}_{ipp'} \mid X_p, X_{p'}] = \Pr(U_{ip} > U_{ip'} \mid X_p, X_{p'}) = \mathbb{E}[Y^{\text{FC}}_{ipp'} \mid X_p, X_{p'}].
$$

\textit{Step 3.} By the experimental design, for any pair $(p, p')$ drawn from the ranking vignette, the attribute vectors $X_p$ and $X_{p'}$ are independently and uniformly randomized. This is the same joint distribution as for any pair of profiles in a forced-choice design.

\textit{Step 4.} Combining Steps 2 and 3: each rank-expanded observation $(\tilde{Y}_{ipp'}, X_p, X_{p'})$ has the same conditional expectation and the same marginal distribution of covariates as a forced-choice observation $(Y^{\text{FC}}_{ipp'}, X_p, X_{p'})$. Therefore, the OLS regression of $\tilde{Y}$ on $X$ identifies the same population parameter as OLS on forced-choice data--namely, the AMCE.
\end{proof}

\subsection{Efficiency of the rank-expanded estimator}

The rank expansion creates $2\binom{K}{2}$ rows per round from $K$ profiles, but these rows are highly dependent. The key insight is that each profile appears in exactly $K - 1$ pairwise comparisons, winning against each profile ranked below it. Its $K - 1$ binary outcomes therefore sum to $K - R_a$, where $R_a$ is its rank. It follows that the $(K-1)$ factor cancels in the OLS normal equations, and the rank-expanded estimator is algebraically equivalent to regressing the normalized rank $\tilde{R}_a = (K - R_a)/(K-1) \in [0, 1]$ on the profile attributes, using one observation per profile.

\begin{proposition}[Variance reduction under rank expansion]
\label{prop:variance}
\label{prop:rankequiv}
The effective sample size of the rank-expanded estimator is $NJK$ profile-level observations (not the $2\binom{K}{2}NJ$ rows in the expanded dataset). Under the null (all $K!$ orderings equally likely), the ratio of asymptotic variances of the AMCE estimator, holding $N$ and $J$ fixed, is:
$$
\frac{\textnormal{Var}(\hat{\beta}_{\textnormal{RCC}})}{\textnormal{Var}(\hat{\beta}_{\textnormal{FC}})} = \frac{2(K+1)}{3K(K-1)}.
$$
This follows from $\textnormal{Var}(\tilde{R}) = (K+1)/[12(K-1)]$ (the variance of a rescaled discrete uniform), compared with $\textnormal{Var}(Y^{\textnormal{FC}}) = 1/4$, and the $K/2$ ratio in effective observations per round. For $K = 2$, the normalized rank reduces to the binary choice indicator, nesting the standard forced-choice estimator as a special case.
\end{proposition}

For the values of $K$ used in this paper, the theoretical variance ratios (and corresponding SE reductions) are:
\begin{center}
\begin{tabular}{cccc}
\toprule
$K$ & Variance ratio & SE ratio & SE reduction \\
\midrule
2 & 1 & 1 & 0\% \\
3 & 4/9 & 2/3 & 33\% \\
4 & 5/18 & 0.53 & 47\% \\
6 & 7/45 & 0.39 & 61\% \\
\bottomrule
\end{tabular}
\end{center}
These figures represent an upper bound on the efficiency gain under the null, ignoring within-subject heterogeneity. In practice, subject-level heterogeneity induces additional intra-cluster correlation that reduces the effective gain, though subject-level clustering of standard errors remains consistent. The simulation results in Section~\ref{sec:sim} and the empirical results in the main text confirm that the observed efficiency gains are broadly consistent with these theoretical predictions.

\subsection{Consequences of IIA violations}
\label{sec:iia_violation}

The proof of Proposition~\ref{prop:unbiased} relies on Assumption~\ref{ass:iia} to marginalize over the remaining profiles. Without IIA, the conditional expectation of the rank-expanded outcome becomes:
$$
\mathbb{E}[\tilde{Y}_{ipp'} \mid X_p, X_{p'}, X_{-(pp')}] = h(X_p, X_{p'}, X_{-(pp')}),
$$
for some conditional expectation function $h$ that cannot, in general, be expressed solely as a function of $(X_p, X_{p'})$.

When we regress $\tilde{Y}_{ipp'}$ on $(X_p, X_{p'})$ alone, we implicitly average over the distribution of $X_{-(pp')}$. By the law of iterated expectations:
$$
\mathbb{E}[\tilde{Y}_{ipp'} \mid X_p, X_{p'}] = \mathbb{E}_{X_{-(pp')}}[h(X_p, X_{p'}, X_{-(pp')}) \mid X_p, X_{p'}].
$$
Under independent randomization, $X_{-(pp')} \perp\!\!\!\perp (X_p, X_{p'})$, so this simplifies to
$$
\mathbb{E}_{X_{-(pp')}}[h(X_p, X_{p'}, X_{-(pp')})].
$$
If $h$ is additively separable in $(X_p, X_{p'})$ and $X_{-(pp')}$, the OLS coefficient on $X_p$ is still unbiased. However, if $h$ contains interaction terms between the focal pair and the remaining profiles--for example, if the weight a subject places on a candidate's career experience depends on the partisan composition of the other profiles in the vignette--then OLS will not in general recover the AMCE, analogous to an omitted-variable bias.

Importantly, randomization bounds the scope of this bias: because the ``omitted" attributes $X_{-(pp')}$ are independent of $(X_p, X_{p'})$, IIA violations cannot introduce bias through correlation between the included and omitted regressors. Instead, the bias operates through the non-linear interaction between the focal and non-focal profiles' attributes in the function $h$. In practice, this means that mild violations of IIA--where the remaining profiles have small and approximately additive effects on pairwise comparisons--will produce negligible bias.

\section{Pre-analysis plan}
\label{sec:pre-reg}

All pre-registration material can be accessed at 
\url{https://osf.io/fawuy}.

\subsection{Hypotheses}
\label{sec:hypotheses}

Our pre-registered hypotheses, tested across both studies, are:

\begin{itemize}
\item [\textbf{H1}] \textit{Subjects complete surveys more quickly under ranked-choice than under forced-choice designs.}
\item [\textbf{H2}] \textit{Attrition and survey satisficing are lower under ranked-choice than under forced-choice designs.}
\item [\textbf{H3}] \textit{Estimation efficiency is greater under ranked-choice than under forced-choice designs.}
\item [\textbf{H4}] \textit{Accuracy of attribute importance is greater under ranked-choice than under forced-choice designs.}
\end{itemize}

\noindent Then, within Study 2 specifically, we pre-registered:

\begin{itemize}
\item [\textbf{H5}] \textit{More profiles per vignette leads to noisier AMCE estimates.}
\item [\textbf{H6}] \textit{Subjects' behavior is consistent with transitivity.}
\item [\textbf{H7}] \textit{Subjects' behavior is consistent with independence of irrelevant alternatives.} 
\end{itemize}

\subsection{Deviations from the pre-analysis plan}

The empirical analysis was followed as detailed in our pre-registered, pre-analysis plans, except for the following deviations:

\paragraph{Study 1:}
\begin{itemize}
    \item In our discussion of H3 (efficiency), we stated that RCC should yield more ``accurate" estimates. This should have read ``precise". The analysis remains as stated in the plan.
    \item For the explicit test of H4 (latent preferences), our attribute importance score was calculated using the \textit{sum} of AMCEs within each attribute. We amended this to the \textit{mean} of AMCEs in this paper, to take into account that, in the candidate context, different attributes had different numbers of levels.
    \item Under H3 (efficiency) we stated we would estimate the adjusted $R^2$ and ROC AUC scores for the models. The adjusted $R^2$ values are low across all conditions (0.03-0.08) and do not differ meaningfully between designs, consistent with the high residual variance inherent to conjoint experiments. ROC AUC scores are reported in the main text (Table~\ref{tab:roc}).
    \item For the explicit preference comparison (H4), we find that neither design's AMCE-derived importance scores are clearly closer to the self-reported explicit scores: the mean absolute deviation is 5.98 (FCC) versus 6.62 (RCC) in the candidate experiment and 3.10 versus 4.58 in the policy experiment.
\end{itemize}

\paragraph{Study 2:}
\begin{itemize}
    \item H5 (AMCE stability across conditions) was tested using both descriptive comparison and the pre-registered Z-tests. Results are reported in the main text.
    \item The pre-registered exploratory analysis of equivalent-profiles comparison is reported in the main text (Section 5.2).
    \item The pre-registered exploratory analysis of intensity ratings is reported as part of Section 5.4 in the main text.
    \item The pre-registered exploratory analysis of ML prediction models is reported alongside the intensity analysis.
\end{itemize}

\section{Further figures and tables}
\begin{figure}[htbp]
    \centering
    \includegraphics[width = 0.9\textwidth]{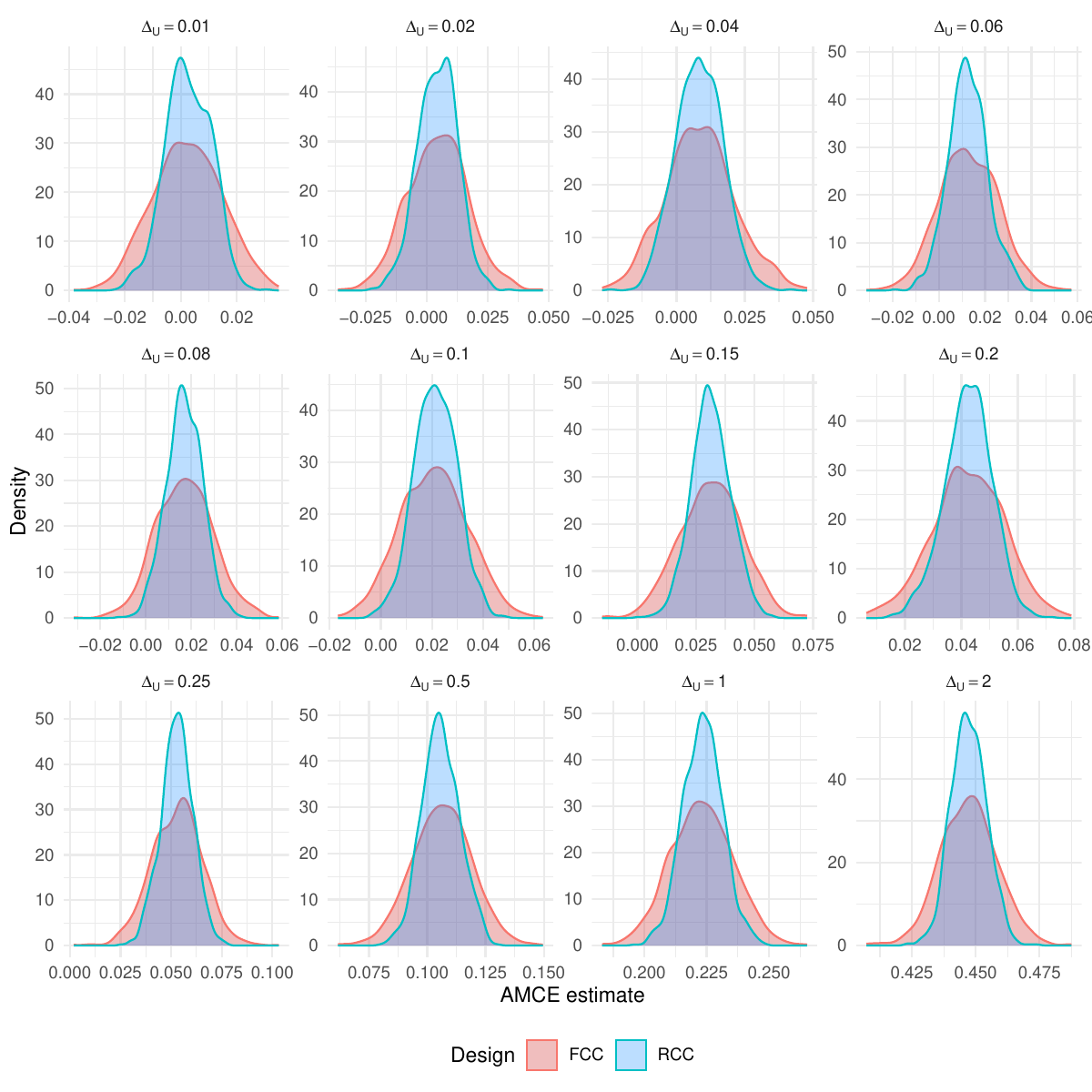}
    \caption{Density of AMCE estimates using RCC and FCC, respectively, over 1000 simulations. Each pane reflects the changes in utility (measured in standard deviations of the residual variance) for showing the same attribute-level.}
    \label{fig:amce_comparison}
\end{figure}

\begin{figure}[htbp]
    \centering
    \includegraphics[width = 0.9\textwidth]{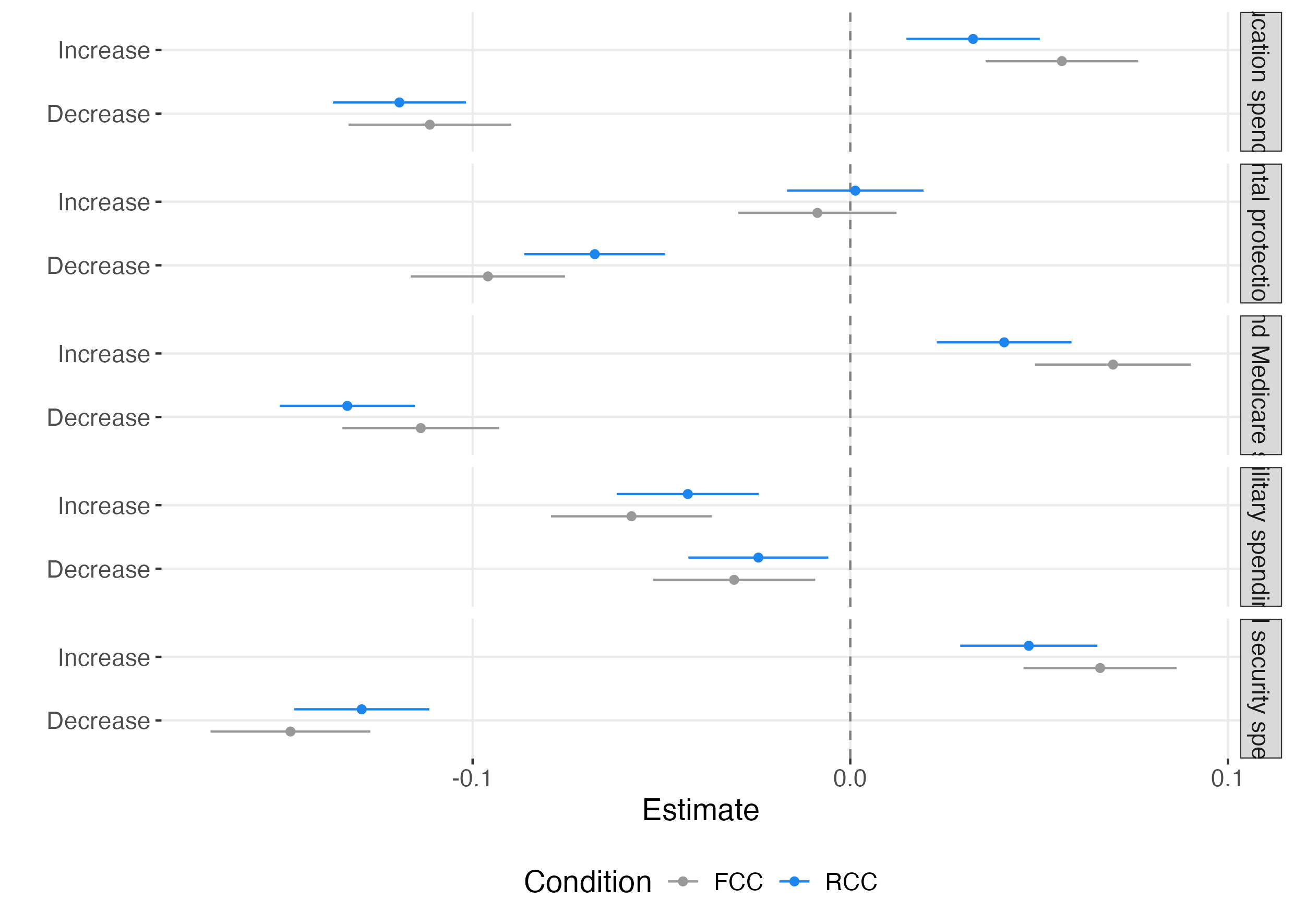}
    \caption{Study 1: Estimated AMCEs from the policy experiment.}
    \label{fig:mod_policy}
\end{figure}

\subsection{Study 2: Additional tables and figures}

\begin{table}
\begin{center}
\begin{tabular}{l c c}
\toprule
 & \multicolumn{1}{c}{Total time taken} & \multicolumn{1}{c}{Time per observation} \\
\cmidrule(lr){2-2} \cmidrule(lr){3-3}
 & Model 1 & Model 2 \\
\midrule
4 profile ranked-choice & $211.20^{***}$ & $-17.88^{***}$ \\
                        & $(13.23)$      & $(0.69)$       \\
6 profile ranked-choice & $408.51^{***}$ & $-21.55^{***}$ \\
                        & $(13.11)$      & $(0.69)$       \\
(Intercept)             & $213.90^{***}$ & $26.74^{***}$  \\
                        & $(9.30)$       & $(0.49)$       \\
\midrule
R$^2$                   & $0.26$         & $0.28$         \\
Adj. R$^2$              & $0.25$         & $0.28$         \\
Num. obs.               & $2837$         & $2837$         \\
\bottomrule
\multicolumn{3}{l}{\scriptsize{$^{***}p<0.001$; $^{**}p<0.01$; $^{*}p<0.05$}}
\end{tabular}
\caption{Effect of ranked choice (compared to forced choice) on completion time}
\label{tab:timings_followup}
\end{center}
\end{table}

\begin{table}
\begin{center}
\begin{tabular}{l c c c}
\toprule
 & \multicolumn{1}{c}{Number of attributes} & \multicolumn{1}{c}{Identified vignette} & \multicolumn{1}{c}{Dimensions} \\
\cmidrule(lr){2-2} \cmidrule(lr){3-3} \cmidrule(lr){4-4}
 & Model 1 & Model 2 & Model 3 \\
\midrule
4 profile ranked-choice & $0.02$        & $0.02$       & $0.00$       \\
                        & $(0.04)$      & $(0.02)$     & $(0.04)$     \\
6 profile ranked-choice & $-0.28^{***}$ & $0.06^{**}$  & $-0.01$      \\
                        & $(0.04)$      & $(0.02)$     & $(0.04)$     \\
(Intercept)             & $0.83^{***}$  & $0.67^{***}$ & $2.51^{***}$ \\
                        & $(0.03)$      & $(0.01)$     & $(0.03)$     \\
\midrule
R$^2$                   & $0.02$        & $0.00$       & $0.00$       \\
Adj. R$^2$              & $0.02$        & $0.00$       & $-0.00$      \\
Num. obs.               & $2837$        & $2837$       & $2837$       \\
\bottomrule
\multicolumn{4}{l}{\scriptsize{$^{***}p<0.001$; $^{**}p<0.01$; $^{*}p<0.05$}}
\end{tabular}
\caption{Effect of ranked choice (compared to forced choice) on manipulation checks}
\label{tab:fmc_followup}
\end{center}
\end{table}

\begin{table}[tbp]
    \centering
    \begin{tabular}{lcccc}
    \toprule
     & \multicolumn{2}{c}{\textbf{ROC Area Under Curve}} & & \\ \cmidrule{2-3}
     \textbf{Experiment} & \textit{FCC} & \textit{RCC} & \textbf{D} & \textbf{$p$-value} \\
     \midrule
     Candidate & 0.571 & 0.627 & 3.293 & $< 0.001$ \\
     Policy & 0.673 & 0.661 & -0.915 & $0.360$ \\
     \bottomrule
    \end{tabular}
    \caption{Study 1: Classification performance on out-of-sample conjoint profiles.}
    \label{tab:roc}
\end{table}

\begin{table}[tbp]
    \centering
    \begin{tabular}{lcc}
    \toprule
     & \multicolumn{2}{c}{\textbf{ROC Area Under Curve}} \\ \cmidrule{2-3}
     \textbf{Condition} & \textit{Fixed profile} & \textit{Random profile} \\
     \midrule
     $K = 2$ & 0.448 & 0.780 \\
     $K = 4$ & 0.627 & 0.794 \\
     $K = 6$ & 0.565 & 0.791 \\
     \bottomrule
    \end{tabular}
    \caption{Study 2: Classification performance on out-of-sample conjoint profiles by number of ranked profiles. * indicates $p < 0.05$ from DeLong's test comparing against the 2-profile baseline.}
    \label{tab:roc_followup}
\end{table}

\begin{table}[tbp]
    \centering
    \begin{tabular}{lrrc}
        \toprule
        \textbf{Specification} & \textbf{Pairs} & \textbf{Respondents} & \textbf{AUC [95\% CI]} \\      
        \midrule 
        \multicolumn{4}{l}{\textit{Panel A: full rank-expanded pairs}} \\
 $K = 2$ (FCC baseline) &  3{,}792 & 948 & 0.703 [0.669, 0.737] \\
 $K = 4$ (all pairs)    & 22{,}122 & 924 & 0.765 [0.735, 0.796] \\
 $K = 6$ (all pairs)    & 57{,}690 & 962 & 0.782 [0.753, 0.811] \\
        \midrule
        \multicolumn{4}{l}{\textit{Panel B: subsampled to 4 pairs per respondent (matched to $K = 2$)}} \\
 $K = 4$, random 1-per-task      & 3{,}687 & 924 & 0.694 [0.660, 0.729] \\
 $K = 4$, top-vs-2nd ranked only & 3{,}687 & 924 & 0.677 [0.642, 0.713] \\
 $K = 6$, random 1-per-task      & 3{,}846 & 962 & 0.706 [0.673, 0.739] \\
 $K = 6$, top-vs-2nd ranked only & 3{,}846 & 962 & 0.617 [0.581, 0.652] \\
        \bottomrule
    \end{tabular}
    \caption{Held-out randomized vignette classification from Study 2, using the structural deep-learning random-utility model of \citet{acharya2026learningpreferencesconjointdata}. Panel A uses all rank-expanded pairs per respondent in each condition; Panel B subsamples each respondent in the $K = 4$ and $K = 6$ conditions to four pairs (matching the per-respondent pair count under $K = 2$) under two schemes: one randomly chosen pair per task, and the pair formed by the top-two-ranked profiles only. Respondent-specific preference vectors $\hat{\boldsymbol{\beta}}(\mathbf{Z}_i)$ are applied to each respondent's held-out vignette profile; AUC measures how well the resulting utility score classifies the binary vote response. All models use $K = 10$ cross-fit folds, 2{,}000 training epochs, and $\mathbf{Z} = \{$gender, partisanship, ideology$\}$.}
    \label{tab:sconjoint}
\end{table}

\begin{figure}[htbp]
    \centering
    \includegraphics[width = 0.9\textwidth]{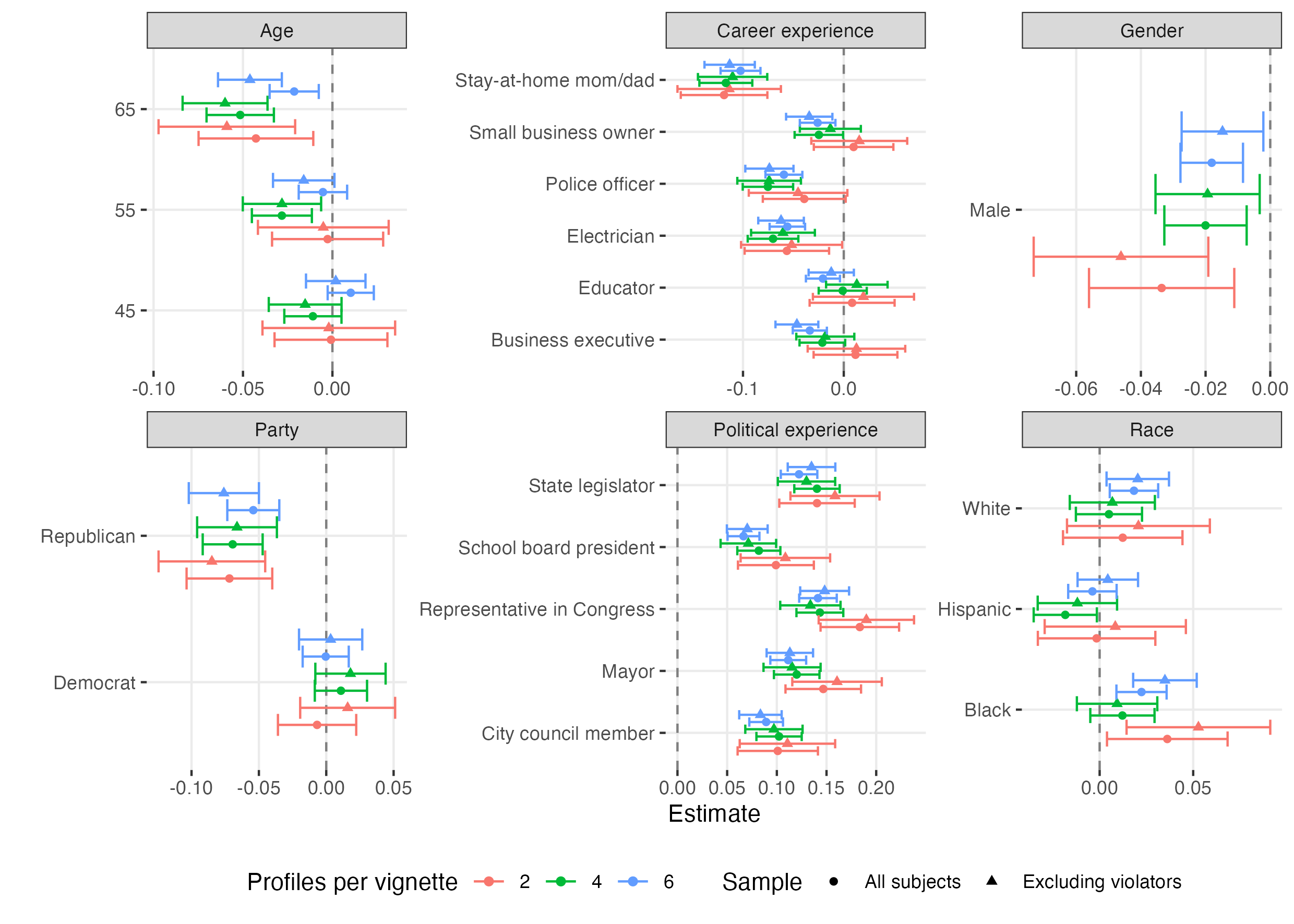}
    \caption{Study 2: AMCE estimates after excluding subjects who exhibited transitivity or IIA violations, by number of profiles per vignette.}
    \label{fig:mod_noviolations}
\end{figure}

\begin{figure}[htbp]
    \centering
    \includegraphics[width = 0.9\textwidth]{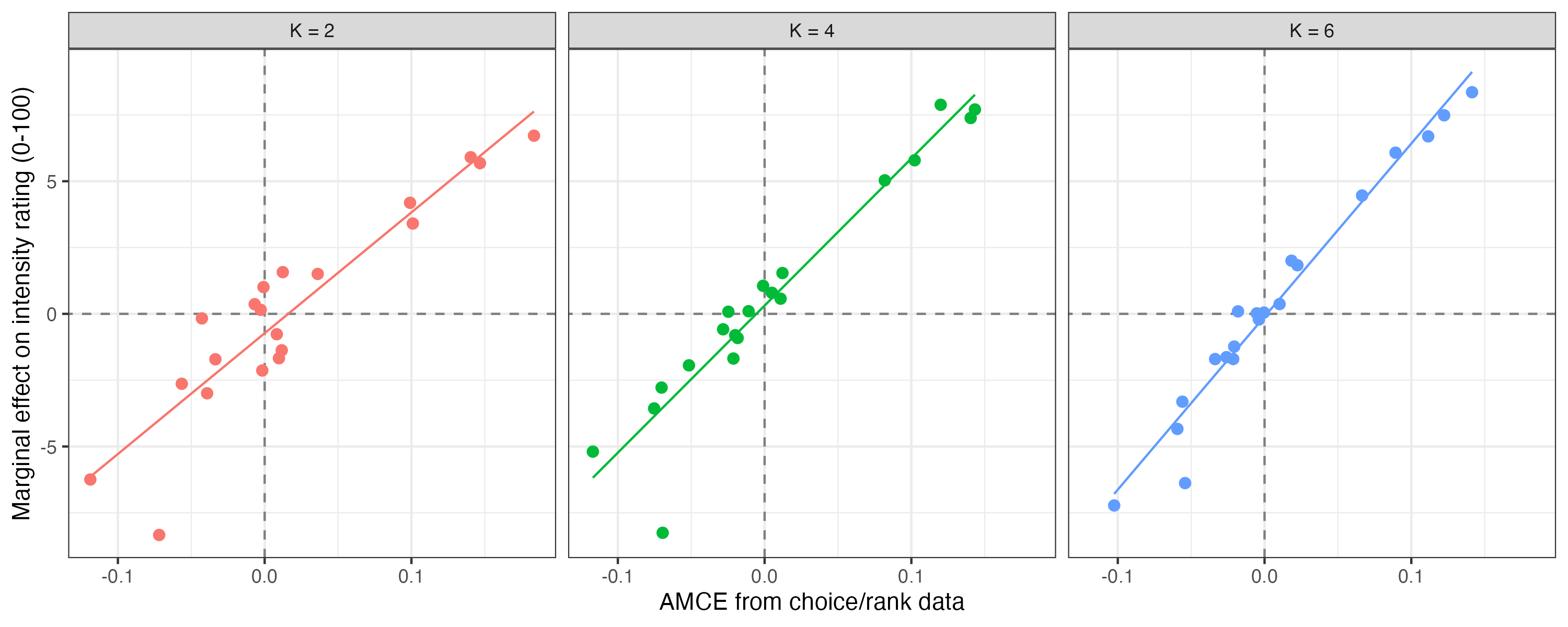}
    \caption{Study 2: Comparison of AMCEs estimated from choice/rank data (x-axis) against marginal effects on intensity ratings (y-axis). Each point is one coefficient. The correlation increases from $r = 0.92$ ($K = 2$) to $r = 0.98$ ($K = 6$).}
    \label{fig:int_vs_amce}
\end{figure}

\begin{figure}[htbp]
    \centering
    \includegraphics[width = 0.65\textwidth]{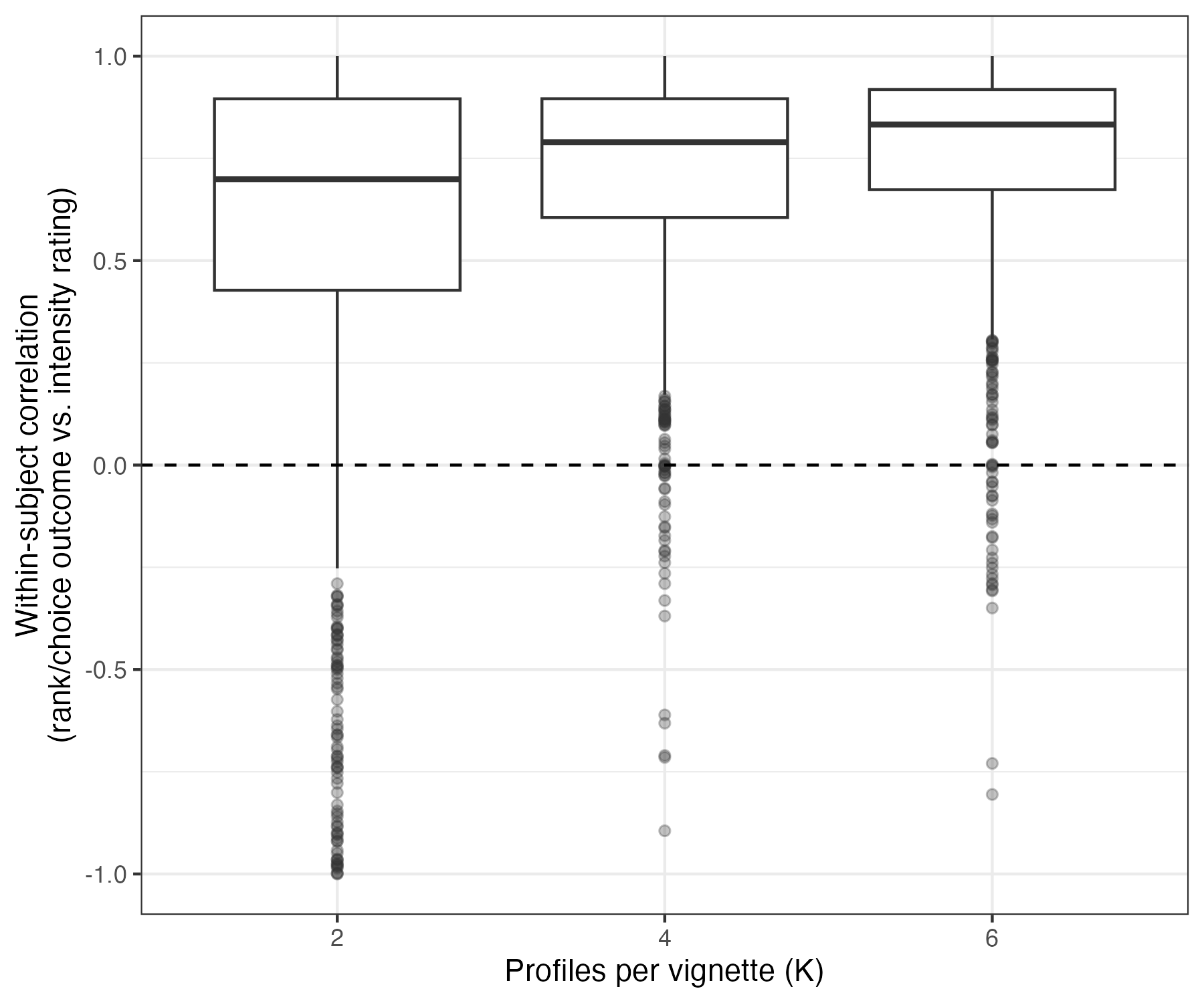}
    \caption{Study 2: Within-subject correlation between normalized rank/choice outcome and intensity rating (0-100 scale), by number of profiles per vignette. The increasing correlation reflects the finer granularity of the rank outcome at higher $K$.}
    \label{fig:rank_intensity}
\end{figure}

\subsection{Covariate balance across conditions}

Tables~\ref{tab:balance1} and~\ref{tab:balance2} report pre-treatment covariate means by experimental condition for Studies 1 and 2, respectively. No significant imbalances are detected across conditions in either study (all $p > 0.05$), confirming successful randomization.

\begin{table}[ht]
\centering
\begin{tabular}{lllr}
  \toprule
Variable & FCC & RCC & p-value \\ 
  \midrule
Age & 46.57 (15.91) & 47.28 (16.02) & 0.29 \\ 
  College degree & 0.57 (0.49) & 0.57 (0.49) & 0.96 \\ 
  Democrat & 0.14 (0.35) & 0.14 (0.35) & 0.87 \\ 
  Female & 0.51 (0.5) & 0.49 (0.5) & 0.24 \\ 
  Ideology (1-7) & 4.92 (2.79) & 5.13 (2.81) & 0.07 \\ 
  Republican & 0.14 (0.35) & 0.15 (0.36) & 0.37 \\ 
   \bottomrule
\end{tabular}
\caption{Study 1: Covariate balance across conditions. Cells report means (SDs). $p$-values from $F$-tests of equality across conditions.} 
\label{tab:balance1}
\end{table}

\begin{table}[ht]
\centering
\begin{tabular}{llllr}
  \toprule
Variable & $K=2$ & $K=4$ & $K=6$ & $p$-value \\ 
  \midrule
Age & 42.26 (13.64) & 43.01 (14.31) & 43 (13.89) & 0.40 \\ 
  College degree & 0.67 (0.47) & 0.69 (0.46) & 0.68 (0.47) & 0.57 \\ 
  Democrat & 0.18 (0.38) & 0.18 (0.38) & 0.17 (0.37) & 0.75 \\ 
  Female & 0.59 (0.49) & 0.58 (0.49) & 0.58 (0.49) & 0.87 \\ 
  Ideology (1-7) & 5.29 (3.08) & 5.18 (3.07) & 5.43 (3.08) & 0.22 \\ 
  Republican & 0.25 (0.43) & 0.23 (0.42) & 0.26 (0.44) & 0.18 \\ 
   \bottomrule
\end{tabular}
\caption{Study 2: Covariate balance across conditions. Cells report means (SDs). $p$-values from $F$-tests of equality across conditions.} 
\label{tab:balance2}
\end{table}

\subsection{Profile position effects}

In forced-choice conjoint experiments, respondents sometimes exhibit a preference for profiles displayed in a particular position \citep[e.g.,][]{Hainmueller2014}. The ranking interface introduces a related concern: respondents using a drag-and-drop interface may be biased toward leaving profiles near their initial display position. To test for this, we regress the rank outcome on profile display position (as a factor), controlling for all conjoint attributes, with standard errors clustered at the subject level.

Table~\ref{tab:position_effects} reports the results. In the forced-choice conditions, we find no significant position effects in either study. In the ranked-choice conditions, however, profiles displayed in later positions receive systematically worse (higher) ranks, consistent with a primacy bias in the ranking interface. In Study 1 ($K = 3$), profiles in position 2 receive ranks approximately 0.08 points worse, and position 3 approximately 0.13 points worse. In Study 2, this effect grows with $K$: at $K = 6$, the position 6 coefficient is 0.70 rank positions. This pattern likely reflects the drag-and-drop interface mechanics, where profiles start in a default order and respondents must actively reorder them. Importantly, including position as a covariate in the AMCE model does not meaningfully change the estimated AMCEs, since profile positions are balanced by the experimental randomization.

\begin{table}[ht]
\centering
\begin{tabular}{llrrr}
  \toprule
Condition & Position & Estimate & SE & $p$-value \\ 
  \midrule
FCC (S1) & Position 2 & 0.01 & 0.01 & 0.65 \\ 
  RCC rank (S1) & Position 2 & -0.08 & 0.02 & 0.00 \\ 
  RCC rank (S1) & Position 3 & -0.13 & 0.02 & 0.00 \\ 
  FCC (S2) & Position 2 & -0.03 & 0.02 & 0.08 \\ 
  RCC rank (S2, K=4) & Position 2 & 0.07 & 0.03 & 0.01 \\ 
  RCC rank (S2, K=4) & Position 3 & 0.15 & 0.03 & 0.00 \\ 
  RCC rank (S2, K=4) & Position 4 & 0.20 & 0.03 & 0.00 \\ 
  RCC rank (S2, K=6) & Position 2 & 0.17 & 0.04 & 0.00 \\ 
  RCC rank (S2, K=6) & Position 3 & 0.32 & 0.04 & 0.00 \\ 
  RCC rank (S2, K=6) & Position 4 & 0.39 & 0.04 & 0.00 \\ 
  RCC rank (S2, K=6) & Position 5 & 0.52 & 0.05 & 0.00 \\ 
  RCC rank (S2, K=6) & Position 6 & 0.70 & 0.05 & 0.00 \\ 
   \bottomrule
\end{tabular}
\caption{Profile display position effects. Coefficients from regressions of the choice/rank outcome on profile position indicators, controlling for all conjoint attributes. Standard errors clustered at the subject level (CR2).} 
\label{tab:position_effects}
\end{table}

\subsection{Predictors of transitivity and IIA violations}

\begin{table}
\begin{center}
\begin{tabular}{l c c}
\toprule
 & Transitivity & IIA \\
\midrule
$K = 4$                    & $0.70^{***}$  & $0.21$        \\
                           & $(0.13)$      & $(0.12)$      \\
$K = 6$                    & $0.75^{***}$  & $0.36^{**}$   \\
                           & $(0.14)$      & $(0.13)$      \\
Age                        & $-0.00$       & $-0.00$       \\
                           & $(0.00)$      & $(0.00)$      \\
Female                     & $0.05$        & $-0.09$       \\
                           & $(0.10)$      & $(0.09)$      \\
College degree             & $0.11$        & $0.20^{*}$    \\
                           & $(0.11)$      & $(0.10)$      \\
Democrat                   & $-0.08$       & $0.23$        \\
                           & $(0.14)$      & $(0.13)$      \\
Republican                 & $-0.08$       & $0.05$        \\
                           & $(0.12)$      & $(0.11)$      \\
Ideology                   & $0.08^{***}$  & $0.07^{***}$  \\
                           & $(0.02)$      & $(0.02)$      \\
Mean task time (s)         & $-0.00$       & $0.00$        \\
                           & $(0.00)$      & $(0.00)$      \\
FMC: attribute count error & $0.02$        & $0.04$        \\
                           & $(0.06)$      & $(0.05)$      \\
FMC: identified vignette   & $-0.14$       & $-0.21^{*}$   \\
                           & $(0.11)$      & $(0.10)$      \\
(Intercept)                & $-2.25^{***}$ & $-1.70^{***}$ \\
                           & $(0.24)$      & $(0.22)$      \\
\midrule
AIC                        & $2690.57$     & $2995.19$     \\
BIC                        & $2761.95$     & $3066.49$     \\
Log Likelihood             & $-1333.28$    & $-1485.60$    \\
Deviance                   & $2666.57$     & $2971.19$     \\
Num. obs.                  & $2831$        & $2812$        \\
\bottomrule
\multicolumn{3}{l}{\scriptsize{$^{***}p<0.001$; $^{**}p<0.01$; $^{*}p<0.05$}}
\end{tabular}
\caption{Logistic regressions predicting transitivity and IIA violations. Coefficients are log-odds. Standard errors in parentheses.}
\label{tab:violation_predictors}
\end{center}
\end{table}

\subsection{Stability of AMCE estimates across rounds}
\label{sec:stability}

To assess whether respondent fatigue or learning affects the AMCE estimates, we estimate separate models for each round using the existing round-by-round analyses reported in Figures~\ref{fig:amce_round_s1} and~\ref{fig:amce_round_s2}. We formally test for instability using Z-tests comparing the first-round and last-round AMCE estimates within each condition.

\begin{sidewaysfigure}
    \centering
    \includegraphics[width = \textwidth]{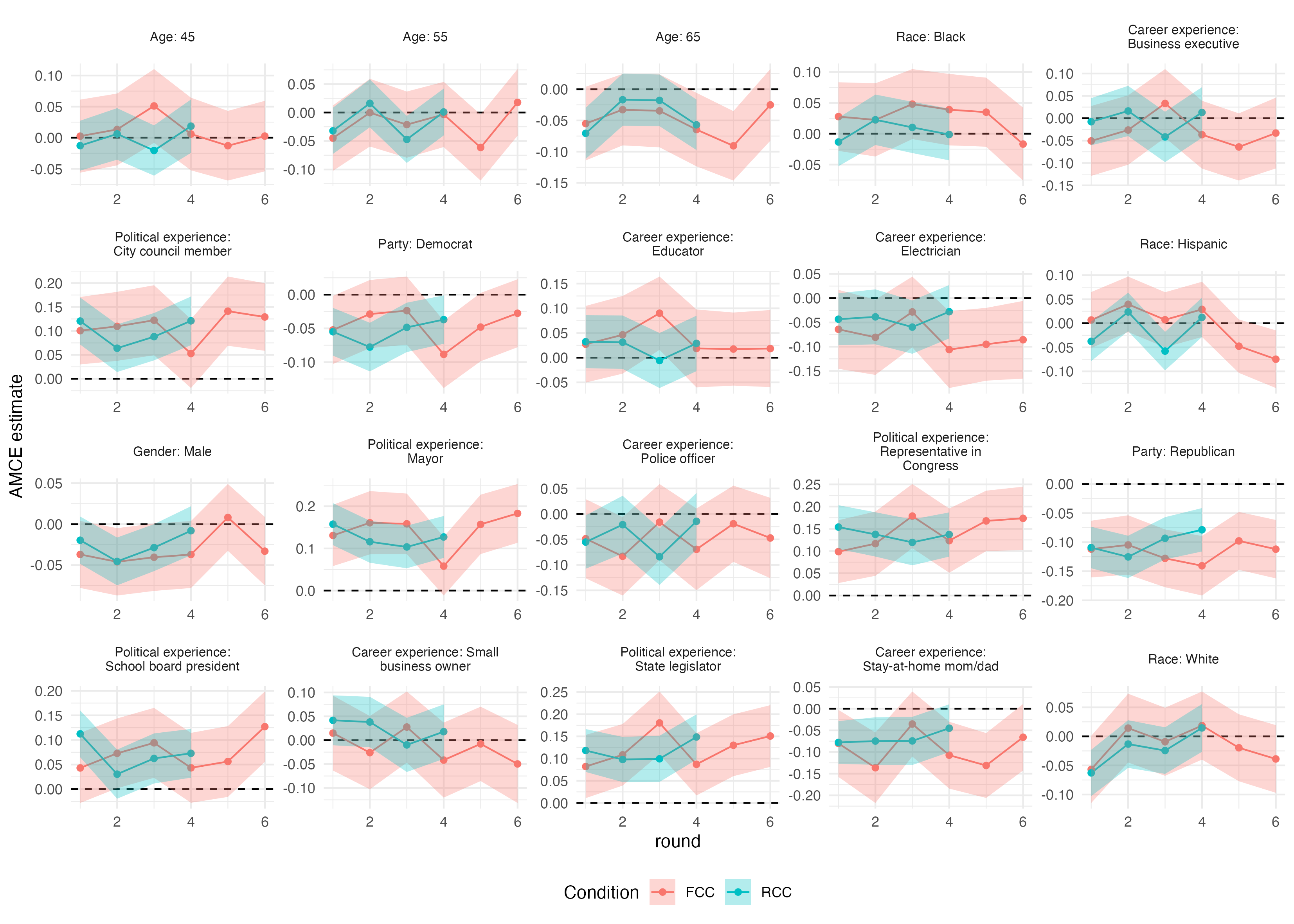}
    \caption{Study 1: Round-by-round AMCE estimates, candidate experiment.}
    \label{fig:amce_round_s1}
\end{sidewaysfigure}

\begin{sidewaysfigure}
    \centering
    \includegraphics[width = \textwidth]{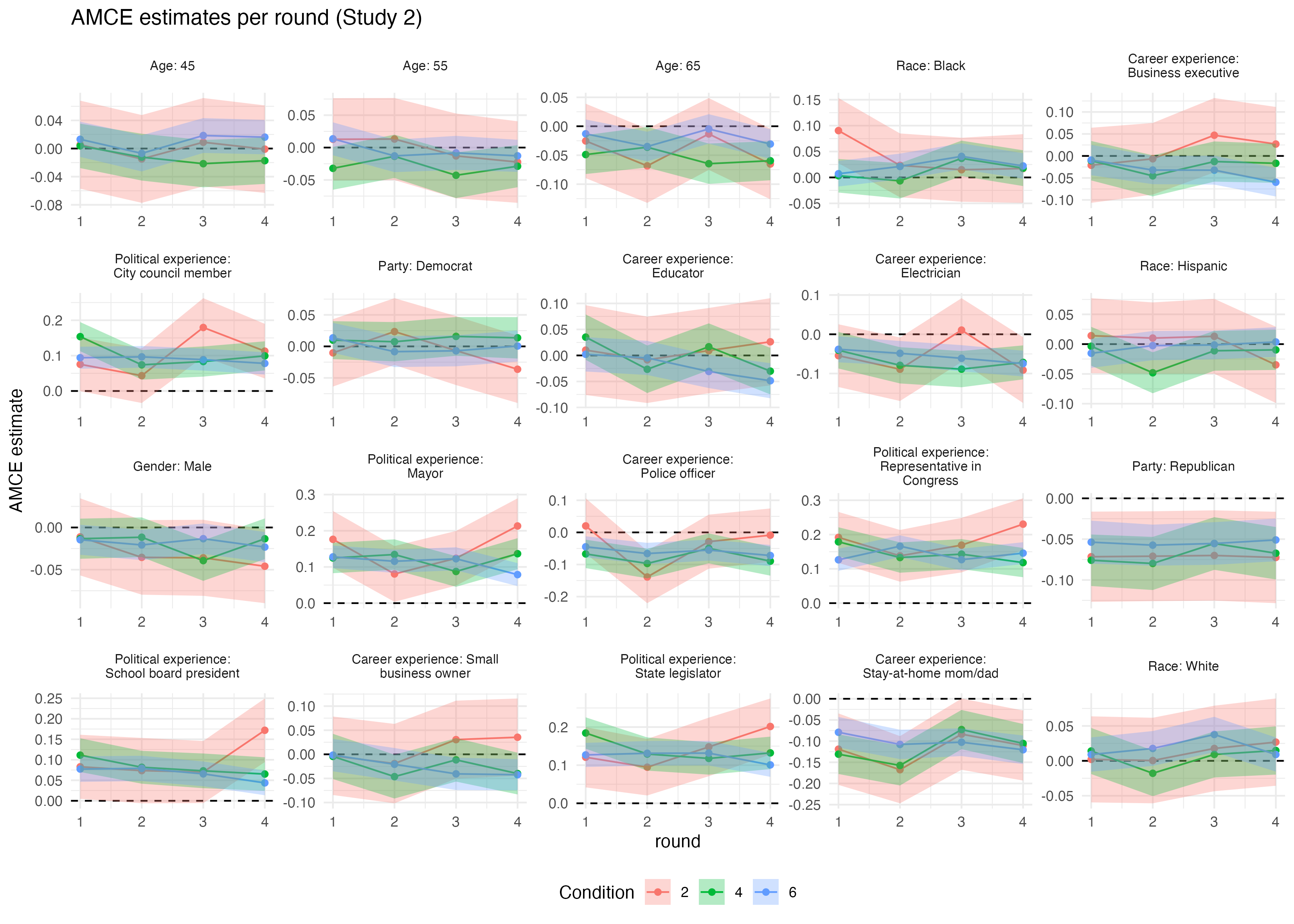}
    \caption{Study 2: Round-by-round AMCE estimates.}
    \label{fig:amce_round_s2}
\end{sidewaysfigure}

Point estimates are noisy at the individual-round level, as expected, but show no systematic trend across rounds in any condition. Of 120 pairwise Z-tests comparing first-round to last-round estimates, 7 are significant at $p < 0.05$--consistent with the 6 expected by chance. This suggests that neither learning nor fatigue systematically biases the AMCE estimates over the course of the experiment.

\subsection{Efficiency gains by respondent characteristics}

A natural concern is whether the efficiency gains of ranking are concentrated among more engaged or cognitively able respondents. To assess this, we split the sample by education (college degree or above versus no college degree) and re-estimate AMCEs within each subgroup. Table~\ref{tab:efficiency_subgroup} reports the mean standard error by condition and subgroup. The SE reduction from ranking is similar across both groups: 13.0\% versus 10.1\% in Study 1, and 41.8\% versus 41.1\% (at $K = 4$) in Study 2. This suggests that the efficiency advantages of ranked-choice designs are not contingent on respondent sophistication.

\begin{table}[htbp]
\centering
\caption{Mean standard errors by education subgroup and condition. SE reduction is relative to the forced-choice baseline within each subgroup.}
\label{tab:efficiency_subgroup}
\begin{tabular}{llccccc}
\toprule
Study & Subgroup & FCC/$K{=}2$ & RCC/$K{=}3$ & SE reduction & & \\
\midrule
Study 1 & College+ & 0.0187 & 0.0163 & 13.0\% & & \\
Study 1 & No college & 0.0215 & 0.0193 & 10.1\% & & \\
\midrule
 & & $K{=}2$ & $K{=}4$ & $K{=}6$ & SE red. ($K{=}4$) & SE red. ($K{=}6$) \\
\midrule
Study 2 & College+ & 0.0222 & 0.0129 & 0.0099 & 41.8\% & 55.1\% \\
Study 2 & No college & 0.0321 & 0.0189 & 0.0143 & 41.1\% & 55.6\% \\
\bottomrule
\end{tabular}
\end{table}

We also examine whether efficiency gains are driven by respondents who spend more time on each task. Splitting respondents by median completion time within each condition, we find that the SE reduction from ranking is similar for fast and slow respondents (Table~\ref{tab:efficiency_speed}): 10.3\% versus 13.3\% in Study 1. In Study 2, fast and slow respondents yield nearly identical standard errors within each condition. This suggests that the precision gains of ranked-choice designs are not solely attributable to more careful respondents taking longer.

\begin{table}[htbp]
\centering
\caption{Mean standard errors by completion speed (median split within condition). SE reduction is relative to the forced-choice baseline within each speed group.}
\label{tab:efficiency_speed}
\begin{tabular}{llccc}
\toprule
Study & Speed group & FCC/$K{=}2$ & RCC/$K{=}3$ & SE reduction \\
\midrule
Study 1 & Fast & 0.0197 & 0.0177 & 10.3\% \\
Study 1 & Slow & 0.0201 & 0.0175 & 13.3\% \\
\midrule
 & & $K{=}2$ & $K{=}4$ & $K{=}6$ \\
\midrule
Study 2 & Fast & 0.0259 & 0.0149 & 0.0114 \\
Study 2 & Slow & 0.0259 & 0.0152 & 0.0116 \\
\bottomrule
\end{tabular}
\end{table}

\subsection{AMCE estimates from intensity ratings}

An alternative to ranking would be to collect continuous preference ratings and estimate AMCEs directly from these. Figure \ref{fig:amce_intensity} plots AMCEs estimated from the 0-100 intensity ratings alongside those from the rank-expanded data in Study 2. The point estimates are similar across approaches (consistent with the high correlations reported in Section 5.4). When placed on a comparable scale, the standard errors are of similar magnitude, reflecting the additional information carried by the continuous outcome. This validates the ranking approach while confirming that the choice-based framework of rank expansion yields comparable precision.

\begin{figure}[htbp]
    \centering
    \includegraphics[width = 0.9\textwidth]{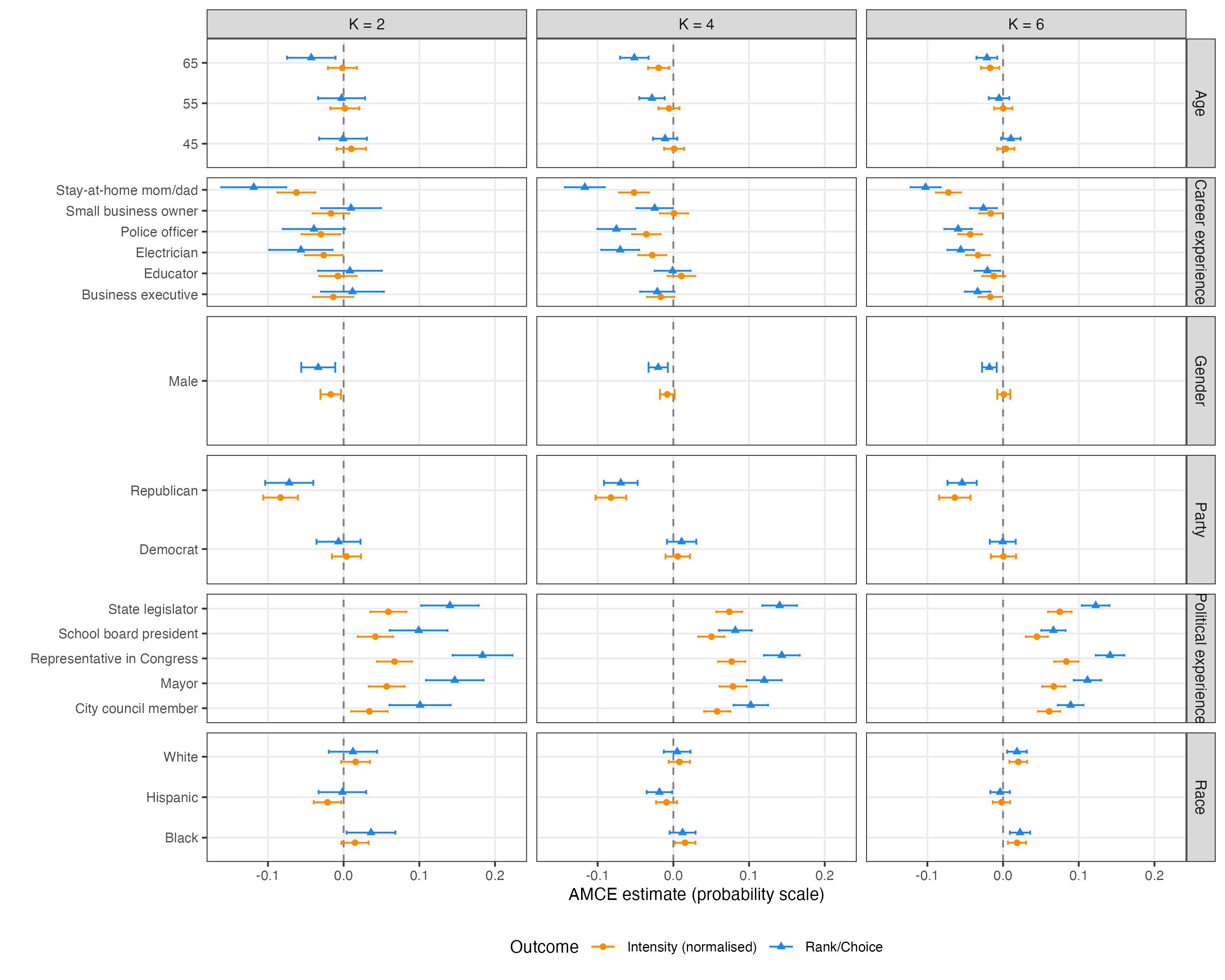}
    \caption{Study 2: AMCE estimates from rank-expanded choice data (blue) versus intensity ratings normalized to the same 0-1 probability scale (orange), by condition. Point estimates are closely aligned across both outcome measures.}
    \label{fig:amce_intensity}
\end{figure}

\subsection{Sensitivity to ranking errors}

To calibrate how robust the AMCE estimates are to ranking errors, we conduct a simulation exercise in which we randomly corrupt a fraction $p$ of pairwise comparisons in the rank-expanded data (from the $K = 4$ condition) and re-estimate the AMCEs. Figure \ref{fig:sensitivity} plots the mean absolute deviation of the corrupted AMCEs from the baseline as a function of $p$, averaged over 200 simulation iterations. At the empirically observed violation rate of approximately 22\%, the deviation is approximately 0.026 percentage points, confirming that the AMCE estimates are robust to the levels of measurement error we observe in practice.

\begin{figure}[htbp]
    \centering
    \includegraphics[width = 0.7\textwidth]{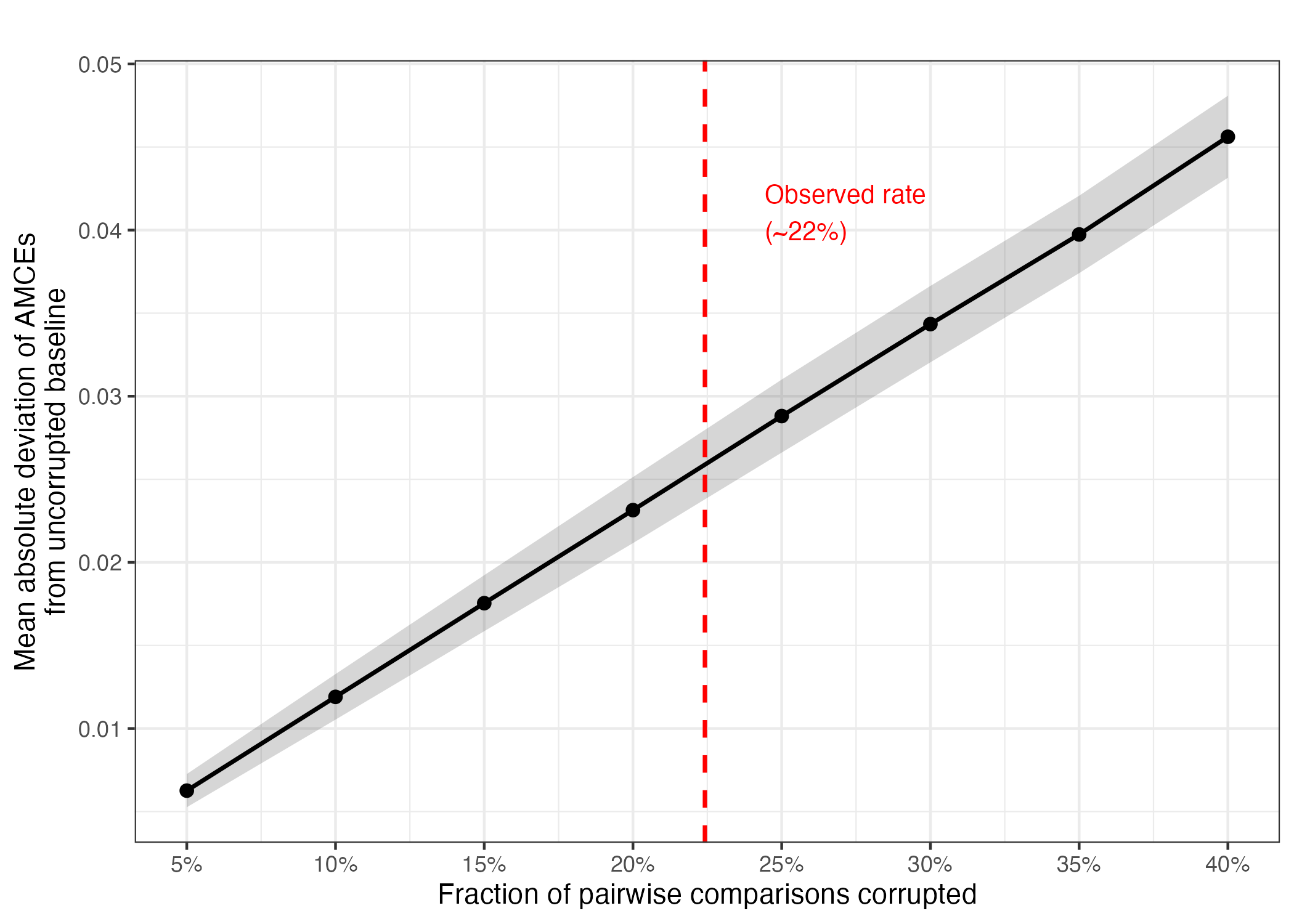}
    \caption{Mean absolute deviation of AMCEs from the uncorrupted baseline as a function of the fraction of corrupted pairwise comparisons. The red dashed line marks the empirically observed violation rate ($\sim$22\%). The shaded region shows $\pm$1 standard deviation across simulations.}
    \label{fig:sensitivity}
\end{figure}

\subsection{Robustness of ML results to model choice}

\begin{table}[htbp]
    \centering
    \begin{tabular}{lc}
    \toprule
    \textbf{Condition} & \textbf{AUC (Logistic)} \\
    \midrule
    FC candidate & 0.475 \\
    RC candidate & 0.484 \\
    FC policy & 0.548 \\
    RC policy & 0.484 \\
    \midrule
    2-profile & 0.610 \\
    4-profile & 0.612 \\
    6-profile & 0.610 \\
    \bottomrule
    \end{tabular}
    \caption{Classification performance using logistic regression instead of XGBoost. AUC scores are computed on the same out-of-sample test data as in Tables~\ref{tab:roc} and~\ref{tab:roc_followup}.}
    \label{tab:roc_logistic}
\end{table}

\subsection{Structural deep-learning modelling}
\label{sec:sconjoint} 

As a robustness check, we also fit the structural deep-learning random-utility model of \citet{acharya2026learningpreferencesconjointdata} separately on the rank-expanded pairs for $K \in \{2, 4, 6\}$ in the Study 2 candidate experiment, and score each respondent's recovered individual preference vector $\hat{\boldsymbol{\beta}}(\mathbf{Z}_i)$ against their held-out randomized vignette.\footnote{The model parameterizes $\hat{\boldsymbol{\beta}}(\mathbf{Z}_i)$ via a neural network of respondent characteristics, with a double/debiased machine-learning correction for valid inference on population averages. We use the production configuration from \citet{acharya2026learningpreferencesconjointdata}: 10 cross-fit folds, 2{,}000 training epochs, and a compact moderator set $\mathbf{Z} = \{$gender, partisanship, ideology$\}$; with a larger $\mathbf{Z}$ the neural network collapses to a near-constant function of respondent characteristics at our sample size.} Using the full set of rank-expanded pairs, held-out AUC rises with $K$, from 0.70 at $K = 2$ to 0.78 at $K = 6$ (Appendix Table \ref{tab:sconjoint}, upper panel). To isolate whether this gain reflects ranking-specific signal in middle-ranked pairs, or simply the larger number of pairwise comparisons that ranking yields, we re-fit the model after subsampling each respondent in the $K = 4$ and $K = 6$ conditions down to four expanded pairs--matching the per-respondent pair count of $K = 2$. The gain disappears: subsampled $K = 4$ and $K = 6$ AUCs return to within confidence intervals of the $K = 2$ baseline (0.69 and 0.71, respectively; lower panel). The estimator's improvement under ranking therefore reflects \textit{pair quantity}, with rank-expanded pairs behaving as equivalent in kind to forced-choice pairs under the random-utility model--consistent with our measurement-level finding that the binarized top-choice signal is equally strong regardless of $K$. This corroborates the efficiency mechanism central to our paper: rankings deliver value by harvesting more within-respondent comparisons, which a structural estimator converts into better-estimated individual preference vectors, not by eliciting a qualitatively different kind of preference data.

\end{document}